\documentclass[10pt]{article}
\usepackage{graphicx,xspace}
\usepackage{multirow}
\usepackage{url}

\usepackage{balance}  

\usepackage{indentfirst}
\usepackage{amsmath}
\usepackage{amssymb}
\usepackage{color}
\usepackage{url}
\usepackage[lined,boxed,ruled,commentsnumbered,linesnumbered]{algorithm2e}
\usepackage[usenames,dvipsnames,svgnames,table]{xcolor}
\usepackage{subfig}
\usepackage{tabularx,booktabs}
\usepackage{paralist}
\usepackage[normalem]{ulem}
\usepackage{threeparttable}
\usepackage{amsthm}

\newcommand{\neha}[1]{}

\newcommand{\reviewText}[1]{}
\newcommand{\eat}[1]{}
\newcommand{\bst}{\textsf{BloomSampleTree}\xspace}

\newtheorem{definition}{Definition}[section]
\newtheorem{proposition}[definition]{Proposition}
\newtheorem{claim}[definition]{Claim}

\def\hatz{\hat{z}}
\def\hatn{\hat{n}}
\def\pr{\mbox{P}}
\def\ex{\mbox{E}}
\def\leafsize{M_\bot}
\def\CB{{\cal B}}
\def\CT{{\cal T}}
\def\CM{{\cal M}}
\def\Node{\mbox{\sf Node}}

\usepackage{ifluatex}
\ifluatex
   \usepackage{fontspec}
\fi

\begin{document}
%
\title{Sampling and Reconstruction Using Bloom Filters}

\author{Neha Sengupta\\
IIT-Delhi\\
neha.sengupta@cse.iitd.ac.in
\and
Amitabha Bagchi\\
IIT-Delhi\\
bagchi@cse.iitd.ac.in
\and
Srikanta Bedathur\\
IBM-IRL\\
sbedathur@in.ibm.com
\and
Maya Ramanath\\
IIT-Delhi\\
ramanath@cse.iitd.ac.in}

\date{}
\maketitle

\begin{abstract}
In this paper, we address the problem of sampling from a set and reconstructing a set stored as a Bloom filter. To the best of our knowledge our work is the first to address this question. We introduce a novel hierarchical data structure called \bst{} that helps us design efficient algorithms to extract an almost uniform sample from the set stored in a Bloom filter and also allows us to reconstruct the set efficiently. In the case where the hash functions used in the Bloom filter implementation are partially invertible, in the sense that it is easy to calculate the set of elements that map to a particular hash value, we propose a second, more space-efficient method called HashInvert for the reconstruction. We study the properties of these two methods both analytically as well as experimentally. We provide bounds on run times for both methods and sample quality for the \bst{} based algorithm, and show through an extensive experimental evaluation that our methods are efficient and effective.
\end{abstract}

\section{Introduction}

Bloom filters, introduced by Bloom in the 1970's~\cite{bloom70}, are space-efficient structures for the set-membership problem. They have found numerous applications in a diverse array of settings because of the tremendous advantages they offer in terms of space. Broder and Mitzenmacher surveyed a host of these applications in 2003~\cite{broder03survey}, and, since then, the usage of Bloom filters has grown and diversified. Typically, these applications rely on the set-membership query being answered correctly with good probability, and are able to deal with the drawback that with some probability a false positive will occur. 

However, one fundamental question has not yet been addressed: How do we sample an element from a set stored in a Bloom filter? A related question -- How do we retrieve the set stored in the Bloom filter? -- has also not been addressed. We believe addressing these two problems will open up the possibility of using Bloom filters in applications that need to store, retrieve and/or sample from a large number of sets.  For example, storing and subsequently sampling from a large number of dynamic, online communities that form on social networks such as Twitter, Flickr, etc. (\cite{Romero_2011}, \cite{ghosh-wsdm:2011}, \cite{cheng-www:2014}), that could help advertisers determine where to target their products. Or storing and retrieving all call records associated with specific locations in crime-related investigations \cite{macmillan-hicss:2013}. 

We note that other compact structures, such as sketches, have been used as compact storage structures from which samples can later be obtained \cite{cormode05, monemizadeh10, jowhari11}. However, a limitation of this approach is that the sketches that are proposed to be created are specifically for the problem of sampling and tend to be output sensitive in their design (and do not support reconstruction). Our work, on the other hand, shows how to draw samples as well as reconstruct sets from a widely-used generic synopsis structure, the Bloom filter, that is also useful for several other applications.

\noindent
\emph{Problem Statement.} Formally speaking, if we are given a set $S$, drawn from a universe or name space $U$, that is stored in a Bloom filter $\CB$ (referred to as the \emph{query Bloom filter}), if we denote by $S(\CB)$ those elements of $U \setminus S$ that are false positives of $\CB$ (i.e., the query ``Is $x \in S$?'' answered by $\CB$ returns YES for all $x \in S \cup S(\CB)$), then:

\begin{enumerate}
\item an algorithm that {\em samples from $\CB$} is one that returns an element chosen uniformly at random from $S \cup S(\CB)$, and,
\item an algorithm that {\em reconstructs the set stored in $\CB$} returns the set $S\cup S(\CB)$.
\end{enumerate}

Since Bloom filters hide information about the elements stored in them providing only (partially correct) answers to membership queries, the natural way of trying to sample from a set stored in a Bloom filter is to fire membership queries with different elements of the name space at the Bloom filter. Such a method, referred to as a Dictionary Attack is not scalable since its running time is linear in the size of the name space, which may be huge.

\noindent
\emph{Solution Overview.} In this paper, we outline a method that approaches this task much more efficiently. Conceptually, we design a data structure, the \bst, that organizes the \emph{namespace} as a binary search tree. That is, each node of the tree stores a subset of the namespace, but at each level of the tree, the union of these subsets yields the \emph{entire} namespace. While the root of the tree, by itself, stores all the elements of the namespace, each leaf stores only a small subset of this namespace. Once this binary search tree is constructed, the key idea is to now locate only those leaves which potentially contain elements present in the given query Bloom filter $\CB$. This is done by intersecting $\CB$ starting from the root of the search tree and working our way towards the leaves. Entire subtrees are pruned away because they yield empty intersections, thus eliminating large parts of the namespace. Once we identify the relevant leaves, we can efficiently sample or reconstruct the original set using the dictionary attack method explained above. Note that this search tree needs to be constructed only once and will be repeatedly used for different query Bloom filters.

A drawback in this approach is that we are storing the entire namespace in the \bst{}, even though only a small part of it may be actually occupied. Sparse occupancy of a namespace is a regular occurence, especially when we consider non-numeric keys such as strings, where the namespace is typically of the order of $2^{64}$, but the actual occupancy is likely to be of the order of $2^{30}$ (a little over 1 billion) or perhaps less. Therefore, it is space-inefficient to construct a tree for the entire namespace, when a large number of the nodes are going to be empty. In order to address this we present a dynamic version of the \bst{}, we call it Pruned-\bst{} which takes into account the occupancy and can dynamically change its size and structure as the occupancy changes.

The \bst{}-based algorithms we provide for sampling and reconstruction have one very important feature: they do {\em not} require the hash functions used by the Bloom filter to be invertible. Our method only needs to be able to use those hash functions and will work if we are given the implementation of the Bloom filter used to store the set. It is also important to note that we do {\em not} distinguish between true elements of the set stored in the Bloom filter and the false positives that are created in the process of insertion. We approach the Bloom filter as is without any prior knowledge of what has been inserted in it, and without any method of distinguishing true elements from false positives. In summary, our method is designed to work efficiently in a scenario where: i) the namespace is potentially large, even dynamic, ii) the no. of interesting subsets is large (in the millions or billions) and may continue to grow indefinitely, iii) we need to either sample from or reconstruct a subset(s) from the set of interesting subsets, stored in the form of Bloom filters (specifically, these are our query Bloom filters).

We present our methods as an aid to the engineer who has chosen to use Bloom filters for a particular application and has optimised parameters to achieve a given level of accuracy (ie. ratio of true elements to all the elements that return a true answer to a membership query) and who has a way of dealing with false positives.

\paragraph*{Contributions}
(i) We introduce a novel data-structure called \bst{} that can be used sample from a set stored in a Bloom filter as well as reconstruct that set. The \bst{} takes into account the occupancy of the namespace and can change size as the occupancy changes, (ii) We provide theoretical bounds on the runtime and on the quality of samples generated by our \bst{}-based algorithm, show them to be near-uniform. (iii) We show through extensive evaluations that our \bst{}-based algorithms are efficient and provide good quality samples.

\paragraph*{Organization}
In Section~\ref{sec:related} we review the literature. We provide a brief background on Bloom filters and outline the framework in which our methods operate in Section~\ref{sec:background}. Section~\ref{sec:sampling} outlines two baseline techniques for sampling from Bloom filters, along with a discussion on their limitations and the need for our \bst{} method. The \bst{} for sampling and reconstruction are described in detail in Sections~\ref{sec:tree} and \ref{sec:recons}. The results of our detailed experimental analysis are presented in Sections~\ref{sec:expt} and \ref{sec:expt-dynamic}.

\section{Related Work}
\label{sec:related}
Bloom filters are one of the most widely used data structures for approximately  answering set membership queries. Their compact storage and efficient querying through simple bit operators has made them valuable in many different settings. A thorough survey of Bloom filters and their applications are available in ~\cite{broder03survey, tarkoma12survey}.  Despite their widespread use, we are not aware of any work that systematically addresses the problems of generating provably uniform samples using Bloom filters and reconstruct the original set at a given level of accuracy in an efficient way. 

The problem of identifying at least one true positive from the Bloom filter has been considered in an adverserial setting to study how resilient the Bloom filters are for dictionary-based attacks~\cite{bellovin07, naor14}. Given a Bloom filter, an adversary can mount an attack to obtain some elements of the original set by repeatedly posing queries on the Bloom filter -- potentially obtaining a large number of false positives, but also some true positive elements. In our work, we do not operate in an adversarial setting--we assume complete knowledge of the domain of values represented by the Bloom filter and the hash functions used. Given an accuracy level, our aim is to efficiently generate provably uniform random samples from the original set as well as to reconstruct the set as per the accuracy requirements. We systematically solve these problems and back up our solutions with detailed analysis of time complexity and accuracy.

\paragraph*{Sketches for Handling Large Datasets}
Bloom filters belong to a general class of approximation datastructures called \emph{sketches} or \emph{data synopses}, which compactly represent massive volumes of data while preserving some vital properties of the data needed for further analysis~\cite{cormodeFnt}. Some of the sketches used frequently in databases community include histograms, wavelets, samples, frequency and distinct-value based sketches, and so on. 

However, most of these synopses datastructures are used under the assumption that the underlying database is always accessible (e.g., in the case of histograms and samples) or not required (e.g., streaming scenarios).  Reconstructing the underlying set of data values at a given level of accuracy in an efficient manner is not their objective to begin with. Only recently, there have been some results which show how sketches can be used for generating samples, called $L_p$-samplers~\cite{monemizadeh10, jowhari11} which generalize earlier work on inverse sampling~\cite{cormode05}. In these, the goal is to maintain a synopses structure for a stream of updates (i.e., addition and deletion of counts) over a given domain of size $M$, such that at any time it is possible to sample with high accuracy the elements in probability proportional to their number of occurrances. 

Unlike these techniques, our approaches are not focused on streaming setting, and are not designed for specific forms of sampling. The proposed \bst{} approach can be used to generate uniform samples from Bloom filters, a widely-used and generic synopsis structure. 

\paragraph*{Trees and Bloom filters}
In this paper we present the \bst\ that comprises a complete binary tree with Bloom filters stored in every node for the purposes of sampling and reconstruction. Yoon et. al.~\cite{yoon-infocom:2014} also propose a structure that comprises a complete $d$-ary tree with Bloom filters at every node to address the multiset membership problem. Similar in flavor to Yoon et. al.'s structure is {\em Bloofi} proposed by Crainiceanu and Lemire who also address the multiset membership problem by representing each set as a Bloom filter stored at a leaf of a tree, and building the tree by combining these Bloom filters hierarchically~\cite{crainiceanu-infsys:2015}. While the flavour of both these structures is similar to our \bst\ but their concern is the problem of multiset membership testing and so the principle on which their trees are built is completely different from the principle on which we build our tree and the contents of the Bloom filters stored at each node both bear no relationship to what we store in each node.  Another work that combines Bloom filters and trees is by Athanassoulis and Ailamaki~\cite{athanassoulis-vldb:2014} where the authors modify B$^+$-trees by placing Bloom filters at their leaves to create approximate tree indexes that seek to exploit data ordering to improve storage performance. Their structure is completely different from ours in intent and design.

\section{Preliminaries}
\label{sec:background}

In this section, we briefly provide the necessary background in Bloom filters, and subsequently describe the framework which our methods operate in. 
\subsection{Bloom Filters}
A Bloom filter is a probabilistic data structure used to space-efficiently store the elements of a set. It comprises a bit array of $m$ bits, along with $k$ independent hash functions, $h_1 \ldots h_k$. An empty set is represented by a Bloom filter each of whose bits is $0$. For each element $x$ in a non-empty set, the $k$ array positions indicated by $h_1(x) \ldots h_k(x)$ are set to $1$.

A Bloom filter supports membership queries, i.e. a Bloom filter $\CB(S)$ storing a set $S$ can answer queries of the form ``is $x \in S$'' for any $x$, with a false positive probability that depends on the number of bits in $\CB(S)$ and $S$. $x$ is hashed using each of the $k$ hash functions to obtain $k$ array positions. If the bit at each of these positions is set, then the result is positive. Since these bits could have been set due to the insertion of other elements, the probability of a false positive is non-zero and evaluates to $\approx (1 - e^{(-kn/m)})^k$. A Bloom filter is incapable of false negatives.

Other than the membership query, the operation of union and intersection on a pair of Bloom filters is also supported and can be implemented using bitwise $OR$ and $AND$ operations respectively. If $\CB(A)$, $\CB(B)$, and $\CB(A \cup B)$ use the same $m$, the same set of hash functions, and are over the same namespace of values then,
$
\CB( A \cup B)  = \CB(A) \cup \CB(B).
$
Also, if $\CB(A)$, $\CB(B)$, and $\CB(A \cap B)$ use the same $m$ and the same set of hash functions, then
$
\CB(A \cap B) = \CB(A) \cap \CB(B),
$
with probability
$
(1- \frac{1}{m})^{k^2 |A - A \cap B| |B - A \cap B|}
$ \cite{Guo:2010}.

For two fixed, disjoint sets $S_1, S_2 \subset U$, each represented by Bloom filters of $m$ bits and hash functions $h$, the false set overlap predicate $FSO_{\cap}(S_1,S_2,h)$ is true if $\CB(S_1)\cap \CB(S_2) \not= \phi$ even though $S_1 \cap S_2 = \phi$.  A false set overlap of $S_1$ and $S_2$ by Bloom filter intersection of $\CB(S_1)$ and $\CB(S_2)$ is reported with probability \cite{Jeffrey_understandingand},
\begin{equation}
\label{eq:fso}
\pr[FSO_\cap | h] = 1 - \left( 1 - \frac{1}{m} \right)^{k^2|S_1||S_2|}
\end{equation}

\subsection{Framework}
Our methods operate on a database $\mathcal{D} = \{ X_i \,|\, i = 1, \ldots \}$ of 
$X_i = \{x_j | x_j \in \mathcal{M}\}$  which are subsets of elements drawn from a namespace $\mathcal{M}$ of size $M$. Instead of operating on the $\mathcal{D}$ directly,we assume we are only given with a compact approximation $\bar{\mathcal{D}}$ where each $X_i$ is represented by a Bloom filter $\CB(X_i)$, for a given length of the filter (in bits), $m$, and the set of hash functions used in its construction, $H$. Such collections of subsets of elements are commonly seen in many application settings including graph databases  -- to represent the adjacency list of each vertex, information retrieval  -- to represent the list of documents where a keyword occurs, etc.

The first task we are interested in tackling in this setting is that of generating a random sample  given $\bar{\mathcal{D}}$. Specifically, given information about other parameters used in building this approximation viz., $m$, $H$ and $\mathcal{M}$, we would like to obtain a provably uniform random sample from a given $X \in \mathcal{D}$ - the original database. Since we are operating on an approximate representation, it is also expected that a fixed amount of \emph{inaccuracy} (measured as the probability of sampling an element which is not in $X$) is tolerable and is specified as an input to the system to begin with. It should be noted that this inaccuracy is naturally linked to the probability of false positives in Bloom filters, and thus for a given level of inaccuracy (or accuracy) the Bloom filters used in $\hat{\mathcal{D}}$ have to be designed. The second task, a natural extension of the above, is to reconstruct the original entry $X$ in the database with high accuracy. 	

\section{Sampling and Reconstruction}
\label{sec:sampling}

We describe two approaches to sample an element from a set and reconstruct a set stored as a Bloom filter. The first of these is a simple "dictionary attack"-based method (DictionaryAttack). The second uses the weakly invertible property of certain types of hash functions to do sampling (HashInvert). While both methods can be used to sample from as well as reconstruct a Bloom filter, the DictionaryAttack method suffers from high runtime inefficiencies, while the HashInvert method provides no guarantees on the quality of the sample. We compare our \bst{} algorithm against these two baselines and highlight the advantages and disadvantages of each approach in detail in Section \ref{sec:expt}.

\paragraph*{DictionaryAttack: Sampling with Membership Queries}

The DictionaryAttack algorithm relies on reservoir sampling to guarantee a uniform sample. This is equivalent to reconstructing the input set and sampling an element from it. It proceeds as follows. A membership query is fired on the input set for each element in the namespace. When a positive is reported for an element, that element is retained as the sample with diminishing probability proportional to the size of the set reconstructed so far \cite{vitter}. In particular, if $n'$ is the number of positives reported until now, then the $(n'+1)^{th}$ positive is retained as the sample with probability $1/(n'+1)$. Clearly, the complexity of this algorithm is $O(M)$, where $M$ is the size of the namespace.

Note that it is straightforward to use this method to \emph{reconstruct} the original set.

\paragraph*{HashInvert - Sampling with Invertible Hash Functions}
This method assumes that the hash functions are weakly invertible. A hash function $h$ is weakly invertible if given the value of $h(x)$, one can find a set of values $S$ such that $\forall y \in S, h(y) = h(x)$. An example of a weakly invertible hash function is $h(x) = (ax + b) \% c$, where $a$, $b$, and $c$ are constants. Knowing the namespace $M$, it is straightforward to find a set of elements that all hash to $h(x)$. 
Given a Bloom filter $B$, it exploits the weak invertibility of the hash functions to invert a randomly sampled SET bit $s$ into $k$ candidate sets $P_1(s), P_2(s) \ldots P_k(s)$, each obtained using a different hash function. The $k$ candidate sets are subsequently pruned using the membership queries on the Bloom filter to obtain $P'_1(s),P'_2(s) \ldots, P'_k(s)$. A value sampled uniformly at random from $ \bigcup_{i=1}^{k} P'_i(s) $ is the final sample returned.

\paragraph*{Analysis}
When sampling from the obtained candidate sets is done using a method such as reservoir sampling, the HashInvert method occupies no extra space.

Sampling a set bit takes $\mathcal{O}(m)$ time, where $m$ is the size of the Bloom filter. Once a set bit is chosen, inversion using a hash function takes $\mathcal{O}\left(\frac{M}{m}\right)$ time. The overall time taken for sampling is $\mathcal{O}\left(m + k\frac{M}{m}\right)$.

Note that, in contrast to DictionaryAttack, which provides uniformly random samples, no bounds are given regarding the \emph{quality} of the samples in the case of HashInvert. However, the algorithm can be used to \emph{reconstruct} the original set by exhaustively running the HashInvert algorithm on all set bits of the Bloom filter.

A simple trick gives us more benefits from the HashInvert algorithm. If the Bloom filter is dense, then the number of UNSET bits (0-bits) are potentially less than the number of SET bits. Therefore, instead of inverting the set bits, we can invert the unset bits. This results in a set of elements which are \emph{not present} in the original set. Therefore, the original set can be recovered from a set difference operation.

\section{Bloom Sample Tree}
\label{sec:tree}

In this section we define the \bst{} data structure that will help us sample from and reconstruct a set stored in a Bloom filter. The \bst{} basically organises the entire namespace. Note that the \bst{} is built once and is then used repeatedly to sample from any given query Bloom filter $B$.

\subsection{Definition}
\label{sec:tree:definition}

The \bst{} is a complete binary tree, denoted as $\CT$, with $\log M/\leafsize$ levels, where $\leafsize$ is a threshold whose choice we discuss later in this section. Every node in the \bst{} has a Bloom filter that stores a subset of the namespace. Every level of the tree contains the entire namespace partitioned uniformly amongst the nodes of that level. Hierarchically speaking the organisation is laminar in the sense that the union of the subsets of the namespace stored in two sibling nodes gives us the set stored in their parent node. All the Bloom filters used in the \bst{} have the same parameters -- viz., $m$, the number of bits and, $H$, the set of hash functions, as the Bloom filters used for the sets we are sampling from (or trying to reconstruct). The reason for this is that we will be frequently intersecting the Bloom filter $\CB$ of the set of interest with the Bloom filters $\CB_i$ stored at various nodes in the \bst{}. We now present a more formal definition.

\begin{definition}
Given a namespace $\mathcal{M}$ of size $M$, the size of the Bloom filter $m$, a set of hash functions used for construction $H$ of the form $h : M \rightarrow \{0,1,\ldots, m-1\}$, and an integer parameter $\leafsize < M$ the  {\em \bst{}}, $\CT(M, m,H, \leafsize)$, is a collection of Bloom
  filters 
\[ \left\{ \CB_{i,j} : 0 \leq i \leq \log \frac{M}{\leafsize}, 0 \leq j
< 2^i\right\},\]
such that each of these Bloom filters uses a bit vector of size $m$ and the hash functions $H$, and with the property that the Bloom
filter $\CB_{i,j}$ stores the elements 
\[\left\{\ell: j \cdot \frac{M}{2^i} \leq \ell < (j+1) \cdot \frac{M}{2^i}\right\}.\]
\end{definition}
Note that,
\begin{asparaitem}
\item The collection of Bloom filters forms a tree structure. Since the
  portion of the name space stored in $\CB_{i,j}$ is partitioned equally
  amongst the nodes $\CB_{i+1, 2j}$ and $\CB_{i+1,2j+1}$, $\CB_{i,j}$ is the
  parent of these two nodes in the tree.
\item The leaves of the tree all store sets of size $\leafsize$. The
  namespace is not further subdivided. 

\end{asparaitem}

\begin{figure}[tbp]
    \begin{center}
	\includegraphics[width=0.6\columnwidth,trim=110 300 110 80,clip]{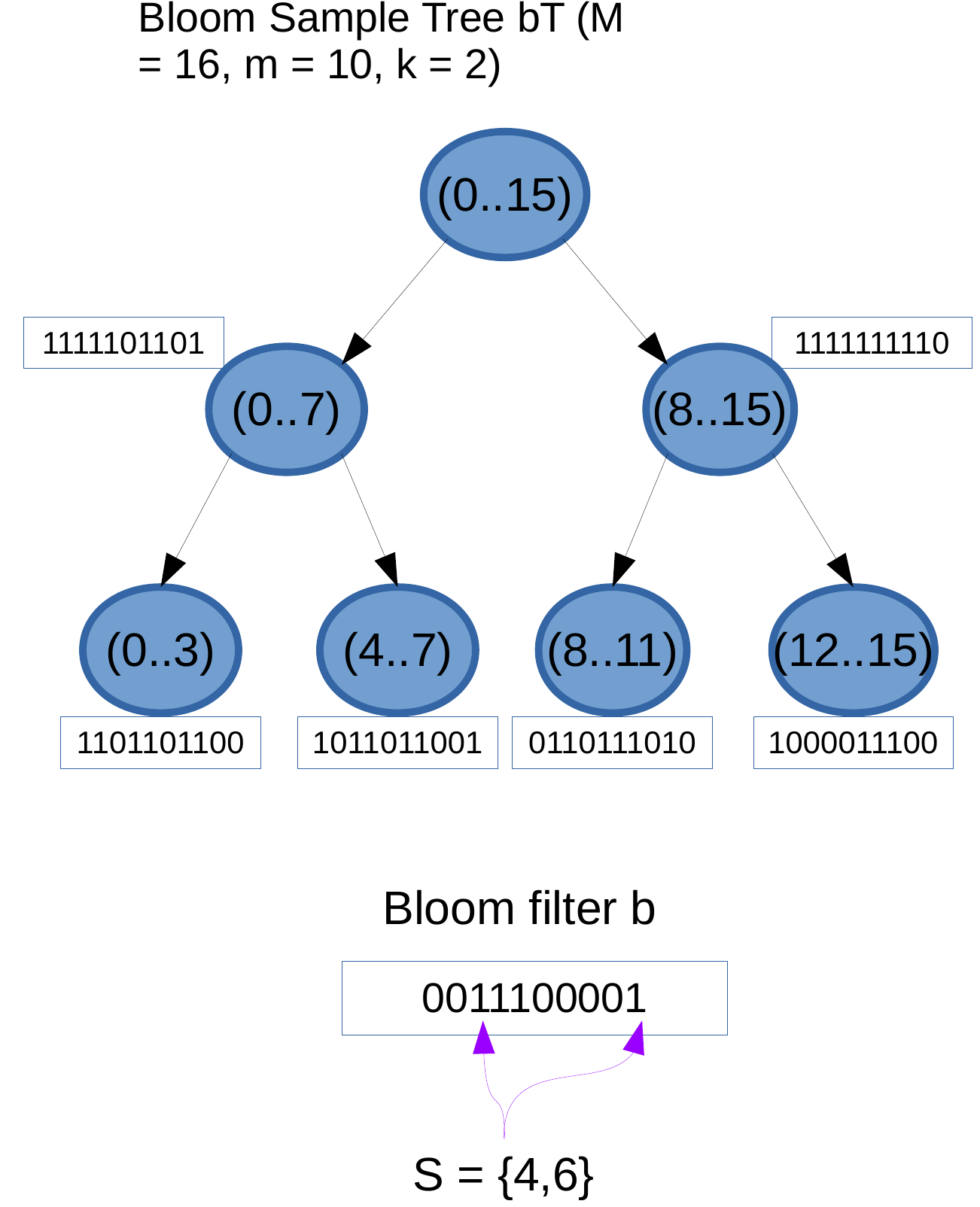}
    \end{center}
    \caption{A \bst{} $\CT$ with 3 levels and the query Bloom filter $b$ representing the set $S$ from which we want to sample}
    \label{fig:BST}
\end{figure}

Figure \ref{fig:BST} shows an example \bst{} for a namespace of $M = 16$. Each node in the tree, except for the root, consist of Bloom filters of size $m = 10$ storing the range of elements depicted at the node. A set $S = \{4,6\}$, stored as Bloom filter $b$ is the query set that we need to sample from or reconstruct. Note that the Bloom filters in $\CT$ are constructed with the same $m$ and $H$ as the Bloom filter $b$ for set $S$.

\subsection{Pruned-\bst{}}
As mentioned in the introduction, even though the namespace itself may be large, it is likely that only a small portion of it is occupied. Therefore, building a complete \bst{} as explained in the previous section potentially wastes a huge amount of space. For example the real-world data set on which we experimented (see Section~\ref{sec:expt-dynamic}) is taken from Twitter and contains 7.2 million user ids distributed in a namespace of size 2.2 billion, i.e., the fraction of the namespace occupied is of the order of $10^{-1}$. Therefore, in practice we build the tree only for those portions of the namespace that are actually occupied; we call this condensed version the ``Pruned-\bst{}''. This tree can dynamically change its structure, based on the change in the occupancy of the namespace. That is, if more of the namespace is assigned, then the tree potentially contains more nodes to reflect that. An overview of the algorithm to build this tree is as follows: Let $\mathcal{M'} \subseteq \CM$ be the set of identifiers that are currently in use ($\CM$ is the namespace, $M' = |\CM'|$ ). 
	\begin{itemize}
		\item Initialise queue with $\Node(0,\log M)$.
		\item Repeat until the queue is empty
\begin{itemize}
		\item Dequeue $\Node(a,b)$.\\ {\small /* $b$ is the level, $a$ is the offset within that level */}
		\item Check if the range $(a,a+2^b-1)$ has a non-empty intersection with $M'$.
\begin{itemize}
		\item If yes, then create $\CB_{b,a/2^b}$ and attach in the tree; insert elements from $\CM'$ in the range $(a,b)$ in $\CB_{b,a/(2^b)}$;   if $b > \log \leafsize$ then enqueue $\Node(a,b-1)$, $\Node((a+2^{b-1}),b-1)$.\\
{\small /* Create the Bloom filter corresponding to the subrange and grow the next level at this point */}
		\item If no, then do nothing.
	\end{itemize}
	\end{itemize}
	\end{itemize}
The above algorithm essentially goes down the tree building subtrees where required to accommodate elements of $M'$ and ignoring subtrees corresponding to ranges that have no overlap with $M'$. Although this algorithm constructs the search tree when the $M'$ is known ahead of time, it is easy to see how to evolve the Pruned-\bst{} when $M'$ grows (e.g. when new Twitter accounts are made)--either we need to insert this new element into already existing nodes in the tree, or we need to create a new node (and potentially its subtree).

The time taken to build the Pruned-\bst{} offline is proportional  to the size of the final tree constructed multiplied by the time for the range query on $M'$. The time taken to update the tree is proportional to the height of the tree.

\subsection{Sampling with the \bst{}}
\label{sec:tree:sampling}

Given a query Bloom filter $\CB$ to sample from, the algorithm proceeds from the root in the following recursive manner and relies on the pruning of the search space for performance gains.

\begin{asparaitem}
\item{At a given (non-leaf) node, compute the intersection of the Bloom filters stored in the left and right children of this node with $b$. If for both child nodes, the intersection is empty, then $b$ does not contain any element belonging in the range associated with this node. Therefore the subtree rooted at this node is pruned from the search.}
\item{If intersection with only one child is non empty, then the search proceeds along that child node. The other child node and the subtree rooted at it are pruned from the search.}
\item{If intersection with both child nodes is non empty, then one of the child nodes is selected with probability directly proportional to the estimated number of elements in their corresponding intersections and the search proceeds along that child. Note that it is possible that the intersection was a false positive and this is discovered further down this subtree. In that case, the search then backtracks and proceeds along the other child node.

The estimated number of elements in the intersection of two Bloom filters $\CB_1$ and $\CB_2$ is given by the following expression \cite{papapetrou10}:
$$
\hat{S}^{-1} ( t_1, t_2, t_{\wedge} ) = 
\frac{
    \ln{\left(  m - \frac{t_{\wedge} \times m - t_1 \times t_2}{m - t_1 - t_2 + t_\wedge} \right) - \ln(m)}
}
{
    k \times \ln \left(1 - \frac{1}{m} \right)
}
$$
where $t_1$ is the number of bits set in $\CB_1$, $t_2$ is the number of bits set in $\CB_2$, $m$ is the size of both Bloom filters, $k$ is number of hash functions used in both, and $t_\wedge$ is the number of bits set in the bitwise AND of $\CB_1$ and $\CB_2$. We recall that Equation~(\ref{eq:fso}) gives us the probability that this intersection is incorrectly estimated to be non-empty when the two sets stored in $\CB_1$ and $\CB_2$ are disjoint. We discuss this issue further in Section \ref{sec:practical}.
}

\item{At a leaf node, every element in the range of the node is checked for membership in $b$. The sample at this leaf node is a value sampled uniformly at random from the set of values that satisfy the membership test of $b$. If none of the elements within this range satisfy the membership query it indicates that the search has reached this leaf node due to a (string of) false set overlap(s). In this case, the sample at this node is $NULL$.}
\end{asparaitem}

This algorithm is called {\tt BSTSample}, and a formal description is in Algorithm \ref{alg:BSTSample}

\begin{algorithm}[h]
\SetAlgoLined
\SetKwFunction{BSTSample}{BSTSample}
  \SetKwProg{myalg}{Algorithm}{}{}
  \myalg{\BSTSample{$bST$:BloomSampleTree, $b$:Query Bloom filter}}{
 $[s,e] \leftarrow bST . range$\;
  \tcc{At a leaf, exhaustively check the interval of the leaf for membership in $b$}
 \eIf{$bST$ is leaf}{
  $S \leftarrow \phi$ \;
  \For {$i$ in $(s,e)$}{
   \If{$i \in b$}{
    $S \leftarrow S \cup \{i\}$ \;
   }
  }
  $x \leftarrow $ uniformly sampled from $S$ \;
  return $x$ \;
  }{
    lFlag = EstimatedSize$(bST.left \cap b)$\;
    rFlag = EstimatedSize$(bST.right \cap b)$\;
    \tcc{If no child intersects, we have reached here due to a false positive and should return $NULL$}
    \uIf{($not$ lFlag $and$ $not$ rFlag)}{
       return $NULL$ \;
    }
    \uElse{
       \tcc{Randomly select one to proceed search along, with probability proportional to estimated number of elements}
       $z$ = uniform$(0,1)$\; 
       \eIf{$z < lFlag/(lFlag+rFlag)$}{
         $sample \leftarrow $ \BSTSample{$bST.left,b$} \;
         \tcc{In case no sample is found from that child, search along the other child - Backtracking}
         \uIf{$sample = NULL$}{
            $sample \leftarrow$ \BSTSample{$bST.right,b$}\;
         }{
            \KwRet $sample$\;
         }
       }{
         $sample \leftarrow $ \BSTSample{$bST.right,b$} \;
         \uIf{$sample = NULL$}{
            $sample \leftarrow$ \BSTSample{$bST.left,b$}\;
         }{
            \KwRet $sample$\;
         }
       }
    }
  }
  }
\caption{Sampling with BloomTrees}
\label{alg:BSTSample}
\end{algorithm}

Figure \ref{fig:BSTDiagram} shows a typical scenario that is encountered when sampling with the \bst{}. The numbers to the side of the node indicate the order in which the nodes are traversed. As shown, the algorithm ultimately generates a sample from a leaf node by following one "true" path out of several false positive paths that may branch out at multiple places. Note, for example, that node 7, is ultimately determined to have led to several false positive paths discovered subsequently in its subtree. In contrast, the whole subtree at node 4 is immediately pruned from the search space. Once the search reaches a leaf, a brute force search is conducted and there is no scope of a false set overlap due to Bloom filter intersection.

\begin{figure}[t!b]
    \begin{center}
      \includegraphics[width=0.35\textwidth,trim=10 300 10 0,clip]{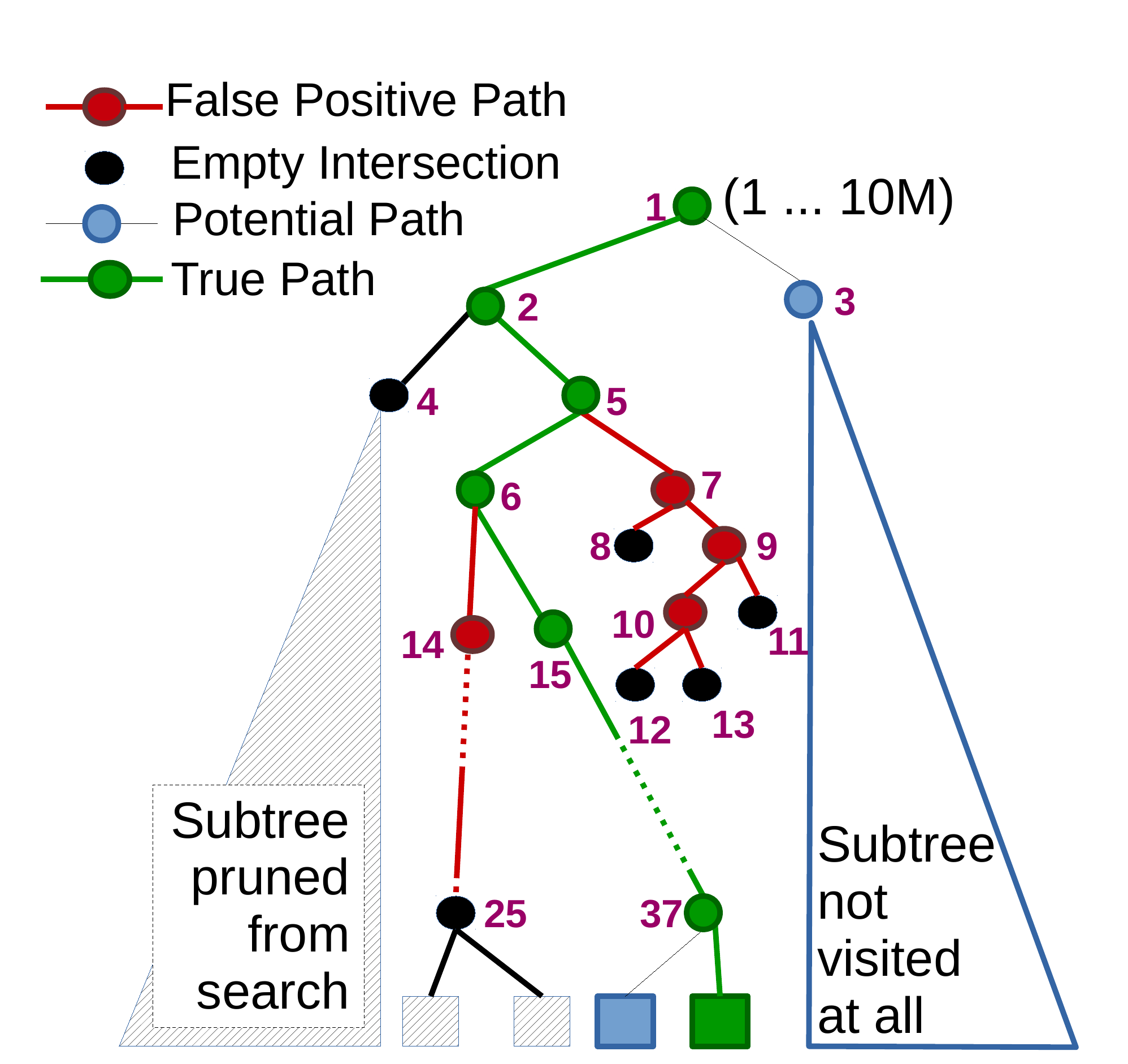}
    \end{center}
    \caption{A typical scenario: Sampling with \bst{}. A False Positive Path is chosen because of errors in determining the empty intersection. The Empty Intersection immediately results in the pruning of a subtree. Potential paths are left unexplored when there is choice of following either subtree. The True Path is the path actually taken by the algorithm to generate the sample.}
    \label{fig:BSTDiagram}
\end{figure}

\paragraph*{Sampling multiple items} The algorithm presented for sampling outputs a single sample. To sample multiple items we could run this algorithm multiple times. However, these multiple runs can be done together in one pass down the \bst{} as we now explain. Given an integer $r$ that is less than the size of the set stored, we send $r$ independent search paths down the \bst{} according to the algorithm {\sf BSTSample}. These paths are sent down the \bst{} in a single pass since all the paths arriving at any internal node or leaf can be processed at the node or leaf before we move on. If at a node we find that the Bloom filters at both its children intersect with the query Bloom filter, we take each of these $r$ paths and, independently of the other paths choose one of the children at random as in {\sf BSTSample} and send the path down to that child. This continues till each of the $r$ paths reaches a leaf. Let us take a concrete example to illustrate this process: Assume we are given a query Bloom filter $\CB$ and $r=3$. We intersect $\CB$ with the Bloom filters $\CB_{1,1}$ and $\CB_{1,2}$ stored in the left and right child of the root of \bst{} and estimate the size of both intersections, let us say they are $k_1$ and $k_2$. Now throw three independent coins biased to come up heads with probability $k_1/(k_1+k_2)$. Suppose two of these coins come up heads and one comes up tails, recursively call two instances of the multiple sampling method with $\CB$ on $\CB_{1,1}$ with $r=2$ and on $\CB_{1,2}$ with $r = 1$. 

It is easy to see that given the tree structure of the \bst{}, such an
extension of the algorithm {\sf BSTSample} will, in general, perform
better than $r$ times the running time for the case when we ask for a
single sample as output. Since all the paths behave like a single
sampling path of {\sf BSTSample} the guarantee on sample quality is
maintained. Finally, if two or more paths happen to reach the same
leaf we can sample at that leaf with or without replacement depending
on whether the $r$ samples are to be generated with or without
replacement.

\subsection{Summary of Analyses}
\label{sec:summary}
Given the \bst\ structure and the algorithm for sampling, we briefly summarise the analyses we performed and the effect of the various parameters.

\begin{asparadesc}
	\item[Quality of samples] The first question we answer is whether our method generates a uniformly random sample. The answer is that a uniformly random sample in indeed generated with high probability. We prove this property in Section \ref{sec:Tree:analysis} and show this empirically as well in Section \ref{sec:expt:sampling}.
	\item[Accuracy] Given that Bloom filters are approximate data structures, it is possible that the samples we generate do not actually belong to the original set (recall that a sample is generated by membership queries at a leaf). We quantify the \emph{accuracy} of our samples as follows:
	
		$$acc = \frac{n}{n+(M-n)*FP}$$
		
		where $n$ is the number of elements in the query set, $M$ is the size of the namespace and $FP$ is the probability of false positives in our Bloom filter implementation.	 The accuracy defined here simply computes the ratio of correct outcomes to all potential outcomes of the algorithm.
		
		Clearly the size of the Bloom filter $m$ has an effect on accuracy and we can determine $m$ based on the desired accuracy. We show the performance of our method for various values of accuracy in Section \ref{sec:expt}.
	\item[Runtime analysis] The runtime of the algorithm depends on the number of false paths it may follow. We analytically show the expected number of nodes visited in Section \ref{sec:Tree:analysis}, given a \bst. However, we also address a practical issue here with regard to runtime -- the cost of performing intersections at a node as opposed to the cost of performing a number of membership queries. Note that it is possible that based on the hash function used, the cost of membership queries may be cheaper or more expensive than the cost of intersections. These two costs are directly related to the no. of elements stored at the leaf, $M_{\perp}$ and the \emph{height} of the \bst is $\log \frac{M}{M_{\perp}}$. We tradeoff the costs as follows:
	
		If $m_{cost}$ is the cost of one membership query to a Bloom filter of size $m$ with $k$ hash functions, and $i_{cost}$ is the cost of an intersection between a pair of bloom filters of size $m$, then, at a current node $N$, storing $N_{\perp}$ values, we would like to determine whether it is better to perform membership queries over $N$ or perform intersections until the leaf which is at most at level $\log N_{\perp}$ below $N$. If performing membership queries is preferred over traversing further down in the tree, we can truncate the tree such that $N$ is the leaf of the tree.
			Hence, we determine
			$ M_{\perp} = \max N_{\perp},$ such that
			$$\frac{N_{\perp}}{\log N_{\perp}} \leq \frac{i_{cost}}{m_{cost}}.$$		
	We empirically show the runtime costs throughout Section \ref{sec:expt}.
	\item[Memory requirements] The memory required to store the \bst\ (which is constructed only \emph{once} and used repeatedly) depends on the size of the Bloom filter $m$ and the number of levels in the tree, $\log \frac{M}{\leafsize}$. An interesting observation here is that, in our framework, there is \emph{no tradeoff} between memory and accuracy or memory and runtime. The tradeoff is between accuracy and runtime, as explained in the previous paragraphs. Therefore, while we set the best possible $m$ and $\log \frac{M}{\leafsize}$ in order to optimize the runtime, the memory required may actually reduce for increased accuracy. The cause for this is that while we will need to use a \emph{larger} Bloom filter in the \bst\ for increased accuracy, we would potentially \emph{reduce} the number of levels so as to reduce the intersection cost (as described in the previous paragraph). The effect of this is that we end up reducing the space used, while increasing accuracy, but also increasing the runtime. We discuss empirical results about this in more detail in Section~\ref{sec:expt:sampling}.
 
\end{asparadesc}

\subsection{Sample quality and running time}
\label{sec:Tree:analysis}

The first question that arises is: what is the distribution of samples
{\tt BSTSample} produces? Our aim is to produce a uniform
distribution from the set stored in the Bloom filter. We present a
theoretical result that shows that the samples produced are near
uniform. We first state the result and then discuss its implications.
\begin{proposition}
\label{prp:uniformity}
Given a set $S$ with $|S| = n$ taken from a name space of size $M$, if
we run {\tt BSTSample} on $S$ with a \bst{}\ $\CT(M,m,H,\leafsize)$ such that $|H| = k$, define 
$\epsilon(m) = \frac{\sqrt{2 nk(\log m + \log \log m + \log
    n)}}{m}.$
Then the 
probability that the sampling algorithm finally samples from an $L
\subseteq S$ of size $\ell$ that is stored in a Bloom filter in a leaf
of the \bst{} lies between
$(1 - \epsilon(m)) \cdot \frac{\ell}{n}$ and $  (1 + \epsilon(m)) \cdot \frac{\ell}{n}$ 
with probability at least $ 1 - \frac{4}{\log m}$,
as long as 
$f(m) = 2\cdot \epsilon(m) \cdot \log \frac{M}{\leafsize} \rightarrow 0$ 
as $m \rightarrow \infty$.
\end{proposition}
\begin{proof}
Probability of a bit being zero after insertion of $n$ elements $ = \left( 1 - \frac{1}{m} \right) ^n$. Let $\hatz$ be a random variable indicating the number of zero bits. We have that,
$$
    \ex(\hatz) = m \left( 1 - \frac{1}{m} \right) ^ {nk}
$$
or that, $\ex(\hatz) = mp$, where $p = \left( 1 - \frac{1}{m} \right)
^ {nk}$. From Theorem 1 of~\cite{Mitzenmacher}, we have that
$$
\pr\left[ | \hatz - mp | > \epsilon m \right] < 2 e^{\frac{- \epsilon^2 m^2}{2nk}}
$$
We set $\epsilon = \frac{\sqrt{2nk(\log m + \log \log m +  \log
    n)}}{m}$. Then $\pr[ |\hatz - mp | > \epsilon m] < \frac{2}{nm
  \log m} = o(1)$.

The estimated size of the population of a bloom filter is $\hatn = \frac{ \log {(\hatz/m)}}{k \log \left( 1 - \frac{1}{m} \right) }$, or $\hatn = c \log (\hatz/m)$, where $c = \frac{1}{k \log \left( 1 - \frac{1}{m} \right)}$. Also, from the above bound, we have that with probability at least $1 - 2/nm\log m$,

\[mp - \epsilon m \leq \hatz \leq mp + \epsilon m.\]
Therefore, with probability at least $1 - 2/nm\log m$,
$$
c \log (p - \epsilon) \leq \hatn \leq c \log (p + \epsilon).
$$
For small values of $\epsilon$,
\[
c \left( 1 - \epsilon \right) \log p \leq \hatn \leq c ( 1 + \epsilon ) \log p.
\]
Substituing $c = \frac{1}{k \log \left( 1 - \frac{1}{m} \right)}$, and $p = \log \left( 1 - \frac{1}{m} \right) ^{nk}$,
$$
n \left( 1 - \epsilon \right) \leq \hatn \leq n \left( 1 + \epsilon \right),
$$
with probability at least $1 - 2/nm\log m$.

Returning to the setting of the BloomTree, let the number of elements in the intersection of the query Bloom filter with the root node be $K$. Similarly, let the number of elements in the intersection of the query Bloom filter with the left and right child of the root node be $K_1$ and $K - K_1$ respectively. In the sampling process, the probability of proceeding along the left child is directly proportional to the estimated number of elements in the intersection of the left child and the query Bloom filter. Let $\hatn_l$ and $\hatn_r$ be the estimated number of elements in the left and right child of the root node respectively.
$$
\pr[\mbox{selecting the left child}] = \frac{\hatn_l}{\hatn_l + \hatn_r}
$$
With probability at least $1 - 4/nm\log m$,
$$
\frac{\left( 1 - \epsilon \right) K_1}{\left( 1 + \epsilon \right) K_1
  + \left( 1 + \epsilon \right) (K - K_1)} \leq \frac{\hatn_l}{\hatn_l + \hatn_r},  
$$
and
\[\frac{\hatn_l}{\hatn_l + \hatn_r} \leq \frac{\left( 1 + \epsilon \right) K_1}{\left( 1 - \epsilon \right) K_1 + \left( 1 - \epsilon \right) (K - K_1)},\]
i.e.,
$$
 \frac{\left( 1 - \epsilon \right) K_1}{\left( 1 + \epsilon \right) K}
 \leq \frac{\hatn_l}{\hatn_l + \hatn_r} \leq  \frac{\left( 1 + \epsilon \right) K_1}{\left( 1 - \epsilon \right) K}.
$$
Now, for $|\epsilon| < 1/2$, $(1-\epsilon)/(1 + \epsilon) \geq 1 -
2\epsilon$ and $(1 + \epsilon)(1 - \epsilon) \leq 1 + 4\epsilon$ and
so 
\[
( 1 - 2\epsilon )\frac{K_1}{K} \leq \frac{\hatn_l}{\hatn_l + \hatn_r}
\leq ( 1 + 4\epsilon )\frac{K_1}{K}
\] 
Consider now a given leaf node $L$ of the Bloom Tree that contains
a subset $S'$ of size $\ell$. The probability that an unbiased process
of sampling should reach this leaf is $\ell/n$. We now estimate what
the probability of reaching this leaf is in the \mbox{\tt BloomTreeSample} method.

Consider the path down the BloomTree to the given leaf: $L_0, L_1,
\ldots L_{\log M/\leafsize} = L$ and the sequence of subsets of $S$ stored in
these nodes be $S = S_0, S_1, \ldots, S_{\log M/\leafsize} = S'$. Then we have
that
\[ \frac{|S_1|}{|S_0|} \cdot\frac{|S_2|}{|S_1|} \cdots \frac{|S_{\log
    M/\leafsize}|}{|S_{\log M/\leafsize -1}|} = \frac{\ell}{n}.\]

Now, using the analysis done above, we can argue that for the choice
of $\epsilon$ above, the probability of moving from $S_i$ to $S_{i+1}$
is at least $(1 - 2 \epsilon) |S_{i+1}|/|S_i|$ and at most $(1 + 4 \epsilon) |S_{i+1}|/|S_i|$ with probability $ 1 -
4/nm\log m$. Repeating this argument over $\log M/\leafsize$ levels we get that 
\[(1 - 2\epsilon)^{\log M/\leafsize} \cdot
\frac{\ell}{n} \leq \pr[\mbox{\mbox{\tt BloomTreeSample} reaches }L],\]
and 
\[\pr[\mbox{\mbox{\tt BloomTreeSample} reaches }L] \leq (1 + 4\epsilon)^{\log M/\leafsize} \cdot
\frac{\ell}{n},\]
with probability at least 
\[ 1 - \frac{4 \log {M/\leafsize}}{nm \log m} \geq  1 - \frac{4}{n \log m}.\]

, where the last inequality applies because $\log{M/\leafsize} < m$ whenever $f(m) \rightarrow 0$ as $m \rightarrow \infty$.
Since the set $S$ can be present in at most $n$ such leaf nodes, hence
for any leaf node containing an element of $S$, the probability that
\mbox{\tt BloomTreeSample} reaches that leaf node is at least 
\[ 1 - \frac{4}{\log m}.\]
Now, we recall that $(1 + x) <
e^x$ and for $x < 1/2$, $e^{-x/2} < (1 - x)$ and so 
In other words, with high probability 
\[e^{-\epsilon \log M/\leafsize} \cdot
\frac{\ell}{n} \leq \pr[\mbox{\mbox{\tt BloomTreeSample} reaches }L] ,\]
and
\[\pr[\mbox{\mbox{\tt BloomTreeSample} reaches }L] \leq e^{4\epsilon \log M/\leafsize} \cdot
\frac{\ell}{n},\]
for any leaf $L$ that contains an element of $S$. 
Now, whenever $\epsilon \log M/\leafsize$ is $o(1)$, i.e., it goes to 0 as $m$
grows, both the bounds go to 1 proving the result.
\end{proof}

\paragraph*{Discussion on sample quality} We note that since $f(m)$ grows faster than $\epsilon(m)$, the condition that $f(m) \rightarrow 0$ as $m \rightarrow  \infty$ implies that $\epsilon(m) \rightarrow 0$ as $m \rightarrow \infty$. 
\eat{The proof of this proposition uses a tail
bound on the number of 1s in a Bloom filter proved by
Mitzenmacher~\cite{Mitzenmacher} to show that the bias estimated by
{\tt BSTSample} is close to the correct value all the way down
the tree.} The import of Proposition~\ref{prp:uniformity} is that the eventual leaf
that {\tt BSTSample} selects for sampling is chosen with
probability that is very close to being proportional to the number of
elements of the set $S$ that belong to the segment of the name space
stored in that leaf. In the absence of false positives, which happens
in the limit as the size of the Bloom filter $m \rightarrow \infty$,
this would lead to perfectly uniform sampling.\eat{ In
Section~\ref{sec:expt} we provide experimental evidence to further
back up our claim that {\tt BSTSample} provides uniform samples
up to a small error.}

We now move on to analysing the running time of the algorithm. Clearly
the running time depends on a number of factors. We provide a
theoretical analysis of the number of nodes that {\tt BSTSample}
visits as it moves down the tree to reach a leaf from where it
generates a sample. Clearly the lower bound for this number is the
height of the tree, i.e., $\log M/\leafsize$. We are not able to match
this lower bound but we are able to control the number of extra nodes
visited and give a result which can guide us on how to choose our
system's parameters to ensure a good asymptotic running time. We show
the following result:
\begin{proposition}
\label{prp:bloom-time}
The expected number of \bst{} nodes visited by the algorithm {\tt
  BSTSample} when sampling from a set of size $n$ using a
\bst{}\  $\CT(M,m,H,M_\perp)$ where $|H| =k$ is
\[O\left( \log \frac{M}{\leafsize} + \frac{M k^2n}{m} \right).\]
\end{proposition}
\begin{proof}[of Proposition~\ref{prp:bloom-time}]
The path {\tt BloomTreeSample} takes down to the leaf node from which
it generates a sample has length equal to the height of the tree,
i.e., $\log (M/\leafsize)$. Along this way, there are many branches
that are caused by false set overlaps, i.e., the intersection Bloom
filter has at least $k$ bits set but {\em no element of the name space
  has all $k$ of its bits set} in this Bloom filter. These branches
need to be followed since it is not possible to distinguish them from
the branch that takes us to genuine true or false positives of the set
stored in the Bloom filter that we are sampling from. In order to
estimate the number of nodes visited over and above the $\log
(M/\leafsize)$ in the true path, we need to estimate how many such
false set overlap nodes we visit. We proceed by showing that below a
certain depth in the BloomTree a false set overlap branch leads to a
constant number of false set overlap nodes being visited below
it. Hence, below that depth the number of extra nodes being visited is
just a constant factor more than the necessary nodes being visited
along the true path. Above that level, however, we cannot make any
such claim so we just assume we visit every node.

Formally, we proceed by making the following claim:
\begin{claim}
\label{clm:half}
Given a node $u$ of depth $d$ in a BloomTree $T$ in which
$\leafsize$ names are stored at every leaf, and a query set $S$ such
that $S\cap M[u] = \emptyset$ where $M[u]$ denotes the subset of the
namespace $M$ that is represented in the subtree of $T$ rooted at $u$,
if the probability that the intersection of two Bloom filters of size
$m$ containing $S$ and $M[u]$, i.e., the quantity
\[\alpha_{S}(d) = 1 - \left( 1 - \frac{1}{m}
\right)^{k^2n|M[u]|},\]
is at most 1/2, and if $L(d)$ is the number of nodes that algorithm
{\tt BloomTreeSample} visits in the subtree of $u$ conditioned on the
event that it reaches
$u$, then 
\[ \ex(L(d)) = \frac{\alpha_S(d)}{1 - 2\alpha_S(d)}.\]
\end{claim}
\begin{proof}
The proof of the claim follows by noting that if $\alpha_{S}(d) = 1 -
x$, then $\alpha_{S}(d - 1) = 1 - \sqrt{x}$, since the number of names
stored in any node of a child is exactly half of those stored by its
parent in the BloomTree. Noting further that it is only if the
sampling algorithm finds an overlap at $u$ that it goes into the
subtrees of $u$'s children, we get
\[\ex(L(d)) = \alpha_S(d) \left[1 + 2 L(d-1)\right] ,\]
Taking expectations repeatedly we get 
\[\ex[L(d)] = (1 - x) + 2(1 - x)(1 - x^{\frac{1}{2}}) +
4(1-x)(1-x^{\frac{1}{2}})(1 - x^{\frac{1}{4}}) \ldots,\]
summed till the leaf level, 
where $x = 1 - \alpha_S(d)$. Since $x < 1$, $1 - \sqrt{x} < 1 - x$,
which allows us to say that 
\[\ex[L(d)] = \frac{1}{2}\sum_{i=1}^{\log(M/\leafsize) - d}
(2\alpha_S(d))^i,\] from where the result follows whenever
$\alpha_S(d) < 1/2$. Note that this argument is basically the same as
saying that the sampling algorithms visits to nodes below a level $d$
for which $\alpha_S(d) < 1/2$ is dominated by a subcritical branching
process, which goes extinct with probability 1 and yields $1/(1-\mu)$
nodes in expectation, where $\mu$ is the mean of the progeny
distribution (c.f.~\cite{athreya:1972})
\end{proof}
Noting that since $|M[u]|$ = $M/2^d$ whenever the depth of $u$ is $d$, the condition that $\alpha_S(d) < 1/2$ is satisfied whenever
\[d > d^* =  \log  \left(\frac{Mk^2n}{m \ln 2}\right),\] and that the
BloomTree up to $d^*$ levels contains $2^{d^*+1} - 1$ nodes, we get
the result.
\end{proof}

\paragraph*{Discussion on running time}

Looking at the result we note that the ratio of name space to Bloom filter size is a critical element in the number of nodes visited, i.e., raising the number of bits in the Bloom filter will benefit the running time (at least up to the point where the second term in the running time analysis continues to dominate the first term). The number of hash functions used and the size of the set being sampled from are also correlated with the running time, which follows intuition.

\subsection{Determining the empty intersection}
\label{sec:practical}
The \bst{} data structure and the algorithm for sampling are both straightforward to implement. However, one practical problem we encounter is that \eat{we do encounter two practical problems. First, }at each node that the algorithm visits, a set intersection needs to be performed that determines whether to prune that branch or not. Unfortunately, there is no reliable way to determine that the size of a set intersection is empty, since even a single set bit results in a non-zero size estimation. Therefore, we use \emph{thresholding} to overcome this problem. That is, if the estimated size of the set intersection is below a particular threshold, we consider the intersection to be NULL. Note that this heuristic can potentially affect the theoretical guarantee offered in Proposition~\ref{prp:uniformity}, but in effect it will not since the probability of making a wrong decision here, i.e. assuming a set if empty when in fact it is not, is very small if we choose the correct threshold. A wrong decision here implies that certain elements of the set are never presented as samples, but, as we see in Section \ref{sec:expt:sampling}, this does not happen in practice.

\section{Reconstruction with \bst{}}
\label{sec:recons}
A recursive traversal of the tree results in a reconstruction of the
set. Given a Bloom filter $\CB$, if the intersection with $\CB$ at a
non-leaf node is empty, then the reconstructed set at this node is the
empty set, and the subtree rooted at this node can be pruned from the
search. However, if the intersection is not empty, then the search
continues along the left and the right children and the final
reconstructed set at this node is the UNION of the reconstructed sets
obtained from the two child nodes.

If the intersection with $\CB$ at a leaf node is not empty, we conduct a
brute force search on the range of this node as before. However,
instead of sampling a value from the set of elements thus obtained, we
return the set itself as the reconstructed set at this node.

\eat{\begin{algorithm}[tb]
\SetAlgoLined
\SetKwFunction{ReconstructionBST}{ReconstructionBST}
  \SetKwProg{myalg}{Algorithm}{}{}
  \myalg{\ReconstructionBST{bT:\bst{},b:Bloom filter}}{
 $[s,e] \leftarrow bT . range$\;
 \eIf{$bT$ is leaf}{
  $S \leftarrow \phi$ \;
  \For {$i$ in $(s,e)$}{
   \If{$i \in b$}{
    $S \leftarrow S \cup \{i\}$ \;
   }
  }
  \KwRet $S$ \;
  }{
    lFlag $\leftarrow$ ($b\cap bT . left$) \;
    rFlag $\leftarrow$ ($b\cap bT . right$) \;
    \uIf{($not$ lFlag $and$ $not$ rFlag)}{
       \KwRet $\phi$ \;
    }
    \uElseIf{($not$ lFlag)}{
       \KwRet \ReconstructionBST{$bT . right,b$}\;
    }
    \uElseIf{($not$ rFlag)}{
       \KwRet \ReconstructionBST{$bT . left,b$}\;
    }
    \uElse{
         $l \leftarrow $ \ReconstructionBST{$bT . left,b$} \;
         $r \leftarrow $ \ReconstructionBST{$bT . right,b$} \;
         \KwRet $l \cup r$ \;
    }
  }
  }
 \caption{Reconstruction with \bst{}}
\label{alg:BloomReconstruct}
\end{algorithm}}
We note that the expected number of nodes of the \bst{} visited by the reconstruction algorithm can be analysed in a manner similar to {\tt BloomTreeSample}. This expected number will come to \[O\left( n \cdot \left(\log \frac{M}{\leafsize} + \frac{\leafsize k^2}{m} \right) \right).\] We note that extracting a single element of a set from the treelike structure of the \bst{} would take $\log M/\leafsize$ and, in the worst case, assuming the different elements of the set are widely distributed in the name space, the worst case number of nodes visited for reconstruction will be $n \log M/\leafsize$ at least. Our algorithm is, unlike in the case of sampling, able to meet this lower bound exactly in an asymptotic sense. Also, since the second term above is directly proportional to $k^2$ and $\leafsize$ and inversely proportional to the size of the Bloom filters used, $m$, we can choose these parameters appropriately to minimize the time taken.
 
\section{Experimental Evaluation with Static Namespace}
\label{sec:expt}
In this section, we describe several experiments that we conducted to determine the effectiveness of our techniques for both sampling as well as reconstruction when the namespace is static. In Section \ref{sec:expt-dynamic}, we describe our experiments where only a fraction of the namespace is actually used.

\subsection{Setup}
\begin{table}
\caption{Parameters of our experiments}
\begin{center}
\scalebox{0.9}{
\begin{tabular}{p{4cm}|p{4.5cm}}
\toprule
{\bf Parameter} & {\bf Range (Default value)}\\
\midrule
Size of the namespace ($M$) & $10^5$ -- $10^7$ ($10^7$)\\
Size of the query set ($n$) & $100$ -- $50,000$ ($1000$)\\
Sampling accuracy($s$)& $0.5$ -- $1.0$ ($0.9$)\\
Hash families & Simple $((ax+b)\%m)$, Murmur3, MD5\\
\bottomrule
\end{tabular}
}
\label{tab:params}
\end{center}
\end{table}
We experimented with both synthetic as well as real datasets. We made extensive use of synthetic datasets to generate controlled micro-benchmarks. We varied the namespace size ($M$), from which the elements are drawn, between $10^5$-$10^7$. We also varied the size of the sets ($n$) between $100$ to $50,000$, and generated them either by uniformly sampling from the namespace or by forming random local clusters (details below). As we pointed out earlier in Section \ref{sec:summary}, the desired accuracy levels can be used to determine the size of the Bloom filter $m$ to construct the \bst. We varied the accuracy requirements between $0.5$ -- $1.0$ and accordingly designed the Bloom filters. For simplicity of experiments we kept the number of hash functions to $3$, although we experimented with different classes of hash functions viz., simple, Murmur3 and MD5. Table~\ref{tab:params} summarizes the parameter choices used in our experiments. Unless mentioned explicitly, the default values of parameters mentioned in this table were used in our experimental evaluations.

\paragraph*{Generating clustered and uniform query sets}
We report results on two kinds of randomly generated query sets. Uniform sets are constructed by generating elements uniformly at random, without replacement, from a given range.

The idea for generating clustered query sets comes from the observations in Web graphs where neighbour sets of vertices typically have their ids clustered around a few nodes~\cite{boldi14}. \eat{ For generating clustered query sets, the following procedure was employed: Start by sampling an element using the uniform distribution. Once an element $s$ is sampled, determine the closest elements $(x,y)$ on either side of $s$ that have not yet been sampled. Divide the probability mass of $s$ equally into $x$ and $y$. Also, add to $x$ and $y$ $0.1$ times the accumulated probability mass of all other remaining elements in order to generate a more aggressively clustered set. }
To generate clustered query sets, $n$ elements were iteratively sampled from the namespace using a $pdf$ that is updated after each sample is drawn. Initially, the $pdf$ begins as that of the uniform distribution. After a sample $s$ is drawn, we identify $x = max\{i|i<s, pdf(i)>0\}$ and $y = min\{i|i>s, pdf(i)>0\}$ as the neighbors of $s$. We divide the $pdf(s)$ equally into $pdf(x)$ and $pdf(y)$, and set $pdf(s)=0$. To generate more aggressively clustered sets, one can subtract $p \%$ from the probability of each element and equally divide the accumulated probability into $x$ and $y$. Here, $p$ controls the degree of clustering. For our experiments, we have used $p = 10$.

\paragraph*{Algorithms}
Our baseline method is the brute-force \emph{dictionary attack} (referred to as $DA$). Additionally, for evaluating the performance of set reconstruction we use \emph{HashInvert} ($HI$). Both these baselines are described in Section~\ref{sec:sampling}. These methods are compared against our \emph{\bst{}} ($BST$) approach. 

\paragraph*{Metrics and Methodology}
We report on the following: 
\begin{asparaitem}
\item \textbf{Number of intersections and set membership operations}: This is our main metric of interest. We compute the depth of \bst{} and the size of the Bloom filters based on accuracy and the relative costs of intersection and membership operations, as discussed in Section \ref{sec:summary}.  For uniform and clustered query sets, we report the average number of intersection and membership operations on Bloom filters over $10,000$ samples. 
\item \textbf{Average time taken}: With the same setup as above, we report on the average time over $10,000$ samples.
\item \textbf{Memory}: We analytically computed the overall size required by the \bst{} based on the size of the Bloom filters used and the tree depth.
\item \textbf{Quality of uniform samples}: We report the $\chi^2$-statistic for the samples generated. In addition, we also show the empirically observed distribution of samples.
\end{asparaitem}

\subsection{Sampling Experiments}
\label{sec:expt:sampling}

\begin{figure}[htb]
\centering
		\subfloat[][$M=10^5$]{
			\includegraphics[width=0.6\columnwidth]{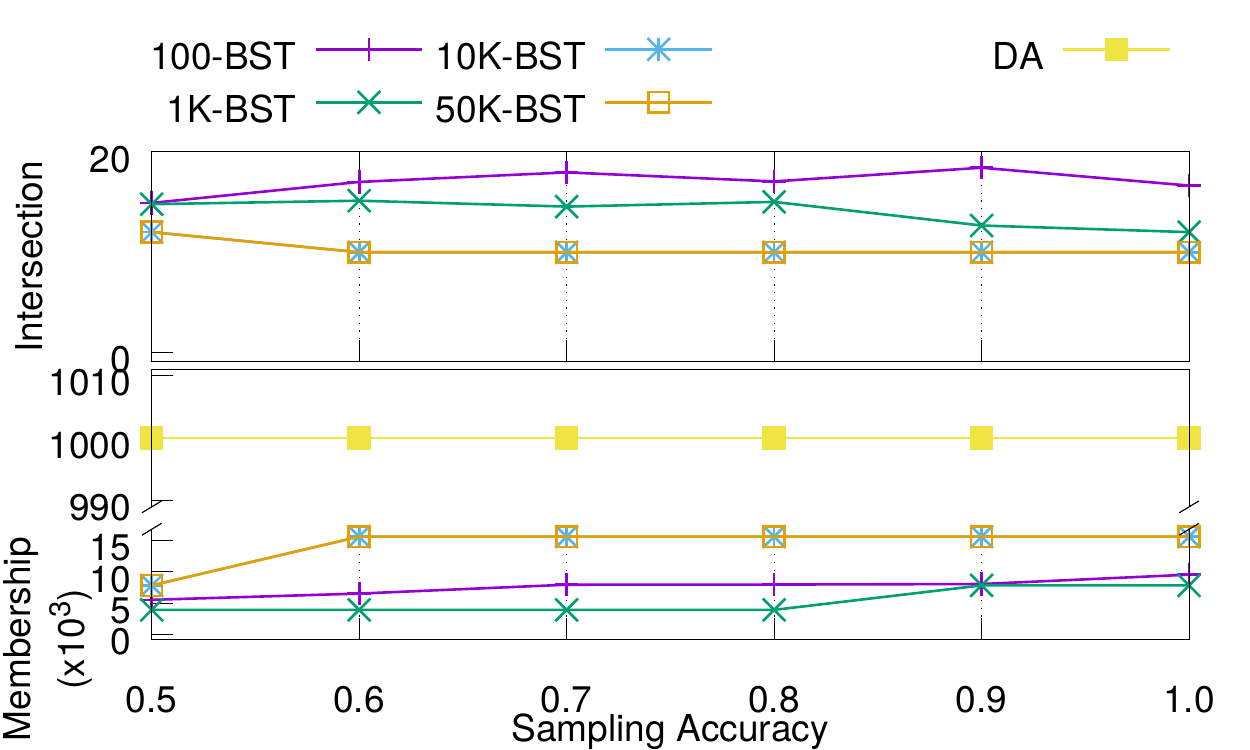}
			\label{fig:S_u_5}
		} \\
		\subfloat[][$M=10^6$]{
			\includegraphics[width=0.6\columnwidth]{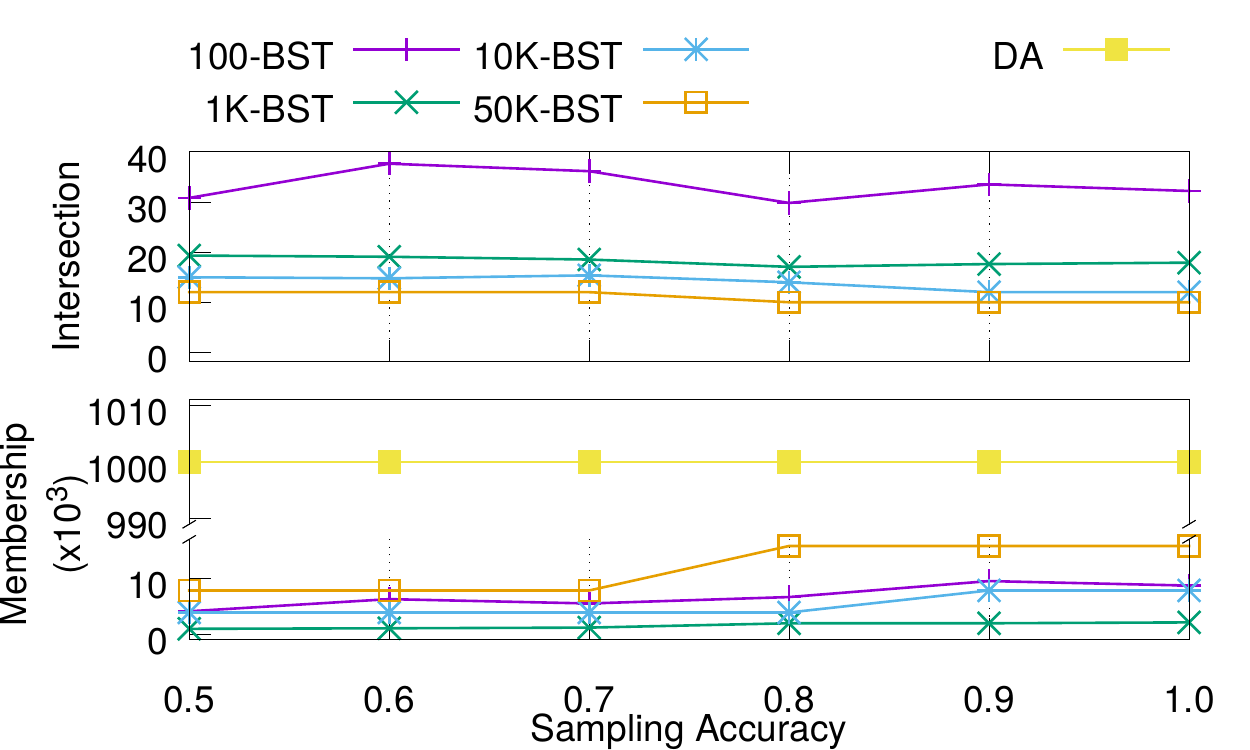}
			\label{fig:S_u_6}
		} \\

		\subfloat[][$M=10^7$]{
			\includegraphics[width=0.6\columnwidth]{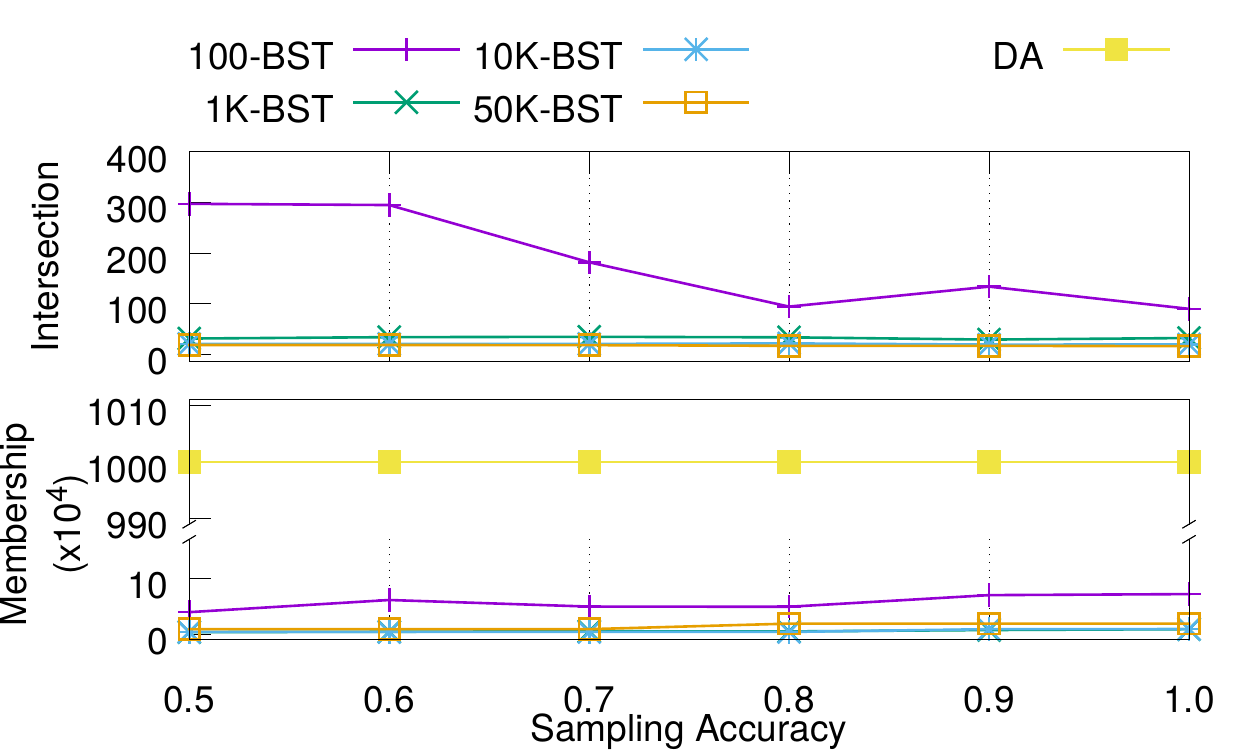}
			\label{fig:S_u_7}
		}
	\caption{No. of intersections and set membership queries for \emph{uniformly random} query sets. $XXX$-BST in the legend refers to the cardinality of the query sets.}
	\label{fig:S_u}
\end{figure}
\begin{figure}[htb]
\centering
		\subfloat[][$M=10^5$]{
			\includegraphics[width=0.6\columnwidth]{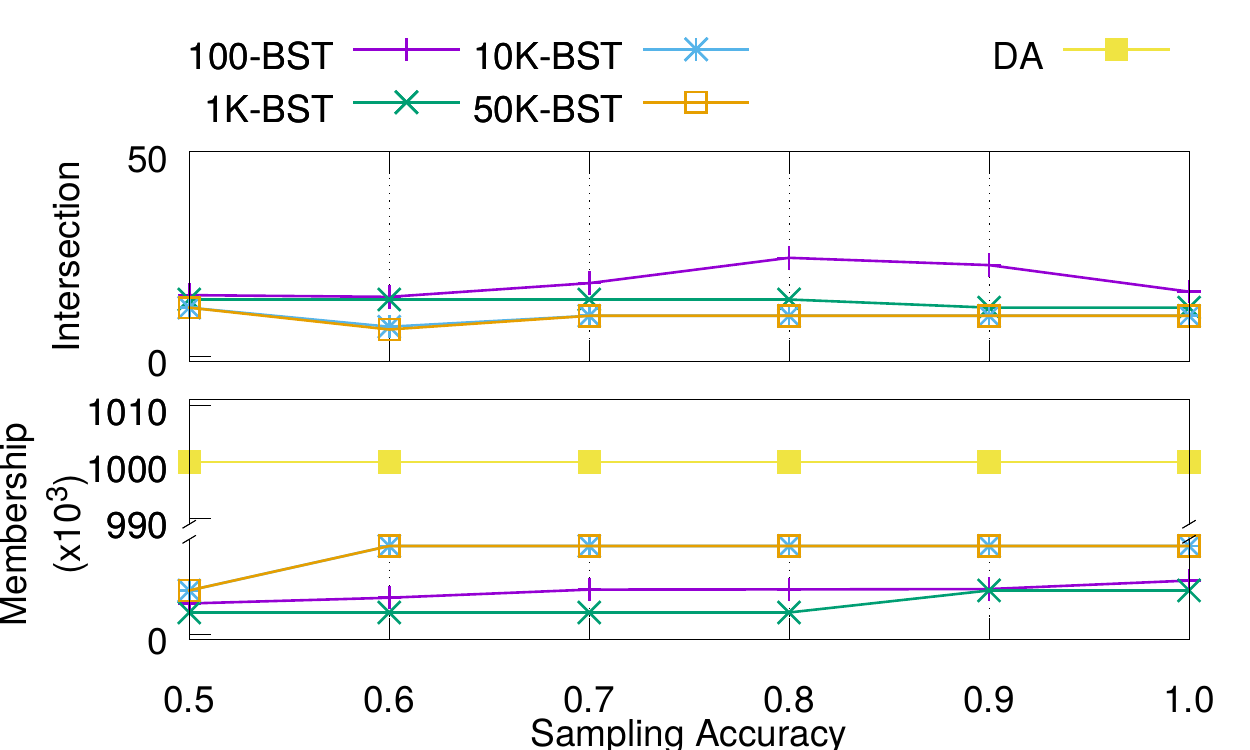}
			\label{fig:S_s_5}
		}\\
		\subfloat[][$M=10^6$]{
			\includegraphics[width=0.6\columnwidth]{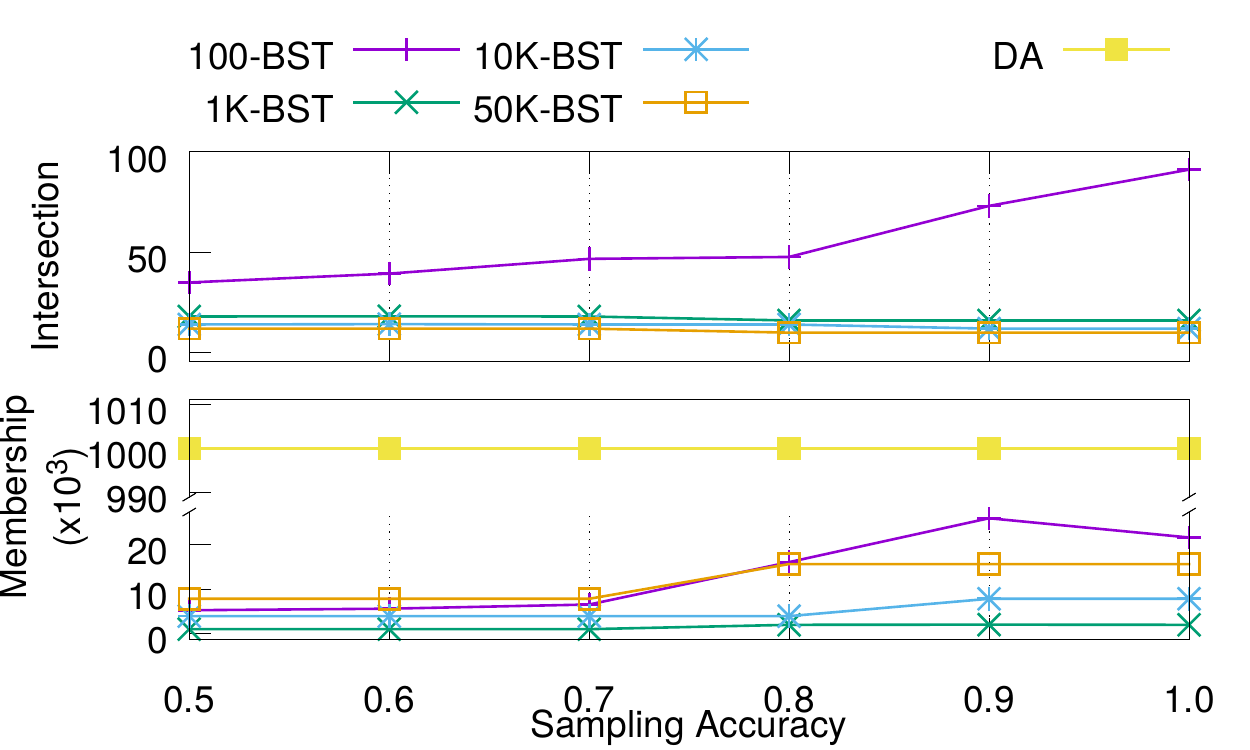}
			\label{fig:S_s_6}
		}\\
		\subfloat[][$M=10^7$]{
			\includegraphics[width=0.6\columnwidth]{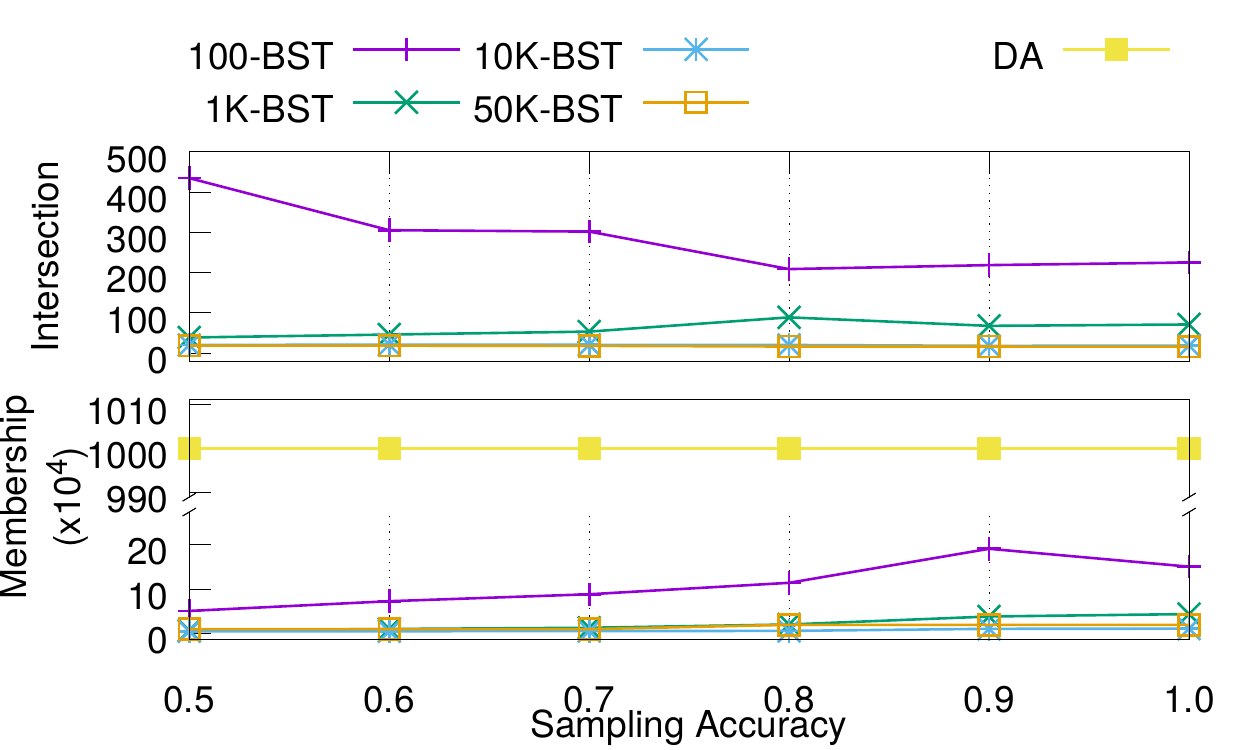}
			\label{fig:S_s_7}
		}
	\caption{No. of intersections and set membership queries for \emph{clustered} query sets. $XXX$-BST in the legend refers to the cardinality of the query sets.}
	\label{fig:S_s}
\end{figure}

Figures \ref{fig:S_u} and \ref{fig:S_s} show the number of intersections and membership operations over \emph{uniformly random} and \emph{clustered} query sets respectively. The DA method always uses $M$ membership operations and no intersection operations. On the other hand, {\bst{}}s try to offset a large number of membership operations with few intersections of Bloom filters. Note that as the sampling accuracy increases, the size of Bloom filters, $m$, increases as well resulting in more expensive intersection operations. 

\paragraph*{Runtime Performance} Both intersection and membership operations become more expensive -- intersections more so than membership operations -- as the Bloom filter size increases. The Bloom filter size, in turn, is determined by the namespace size as well as the accuracy requirements. Thus, the overall efficiency of \bst{} depends on careful balance of the number of these operations. 
	Figure~\ref{fig:S_t_7} shows the average time taken by \bst{} and DA methods, for a namespace size of $10$-million.
As these plots show, {\bst{}}s achieve \neha{$5\times-100\times$} improvements in efficiency over DA for a single sampling round. 

\begin{figure}[htb]
\centering
		\subfloat[][Uniformly random query set]{
			\includegraphics[width=0.6\columnwidth]{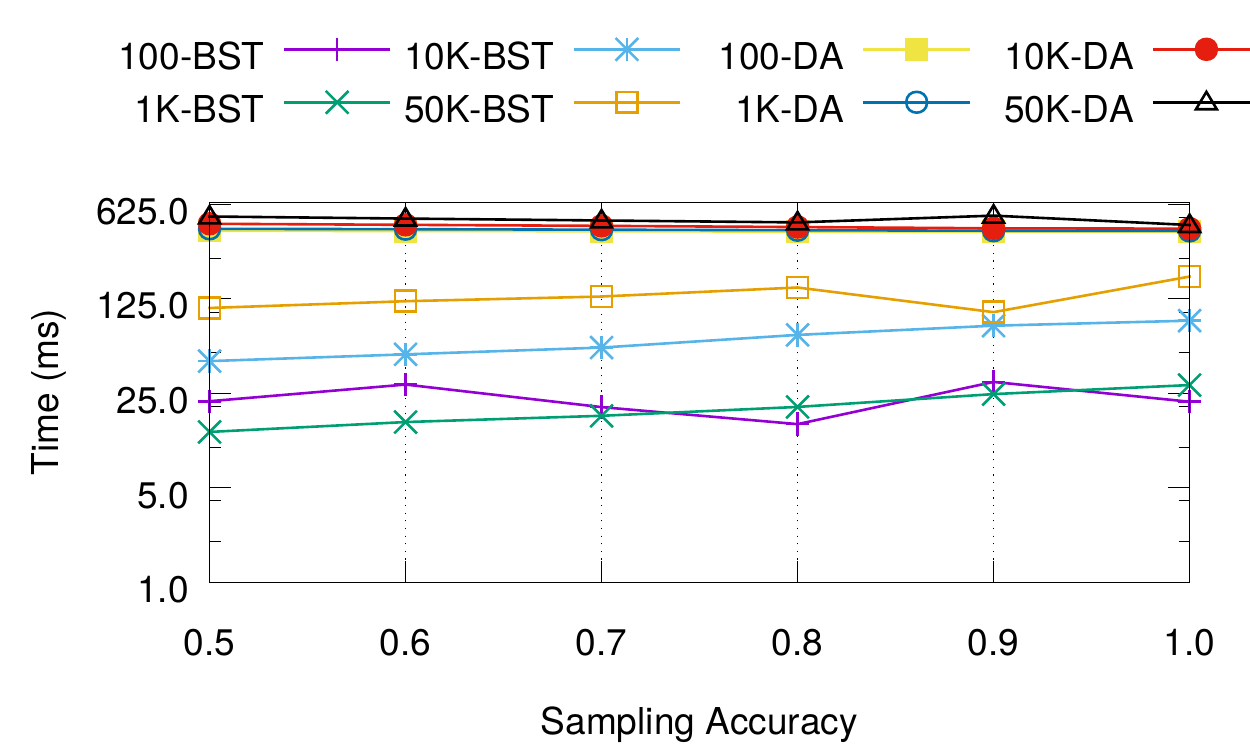}
			\label{fig:S_tu_7}
		}\\
		\subfloat[][Clustered query set]{
			\includegraphics[width=0.6\columnwidth]{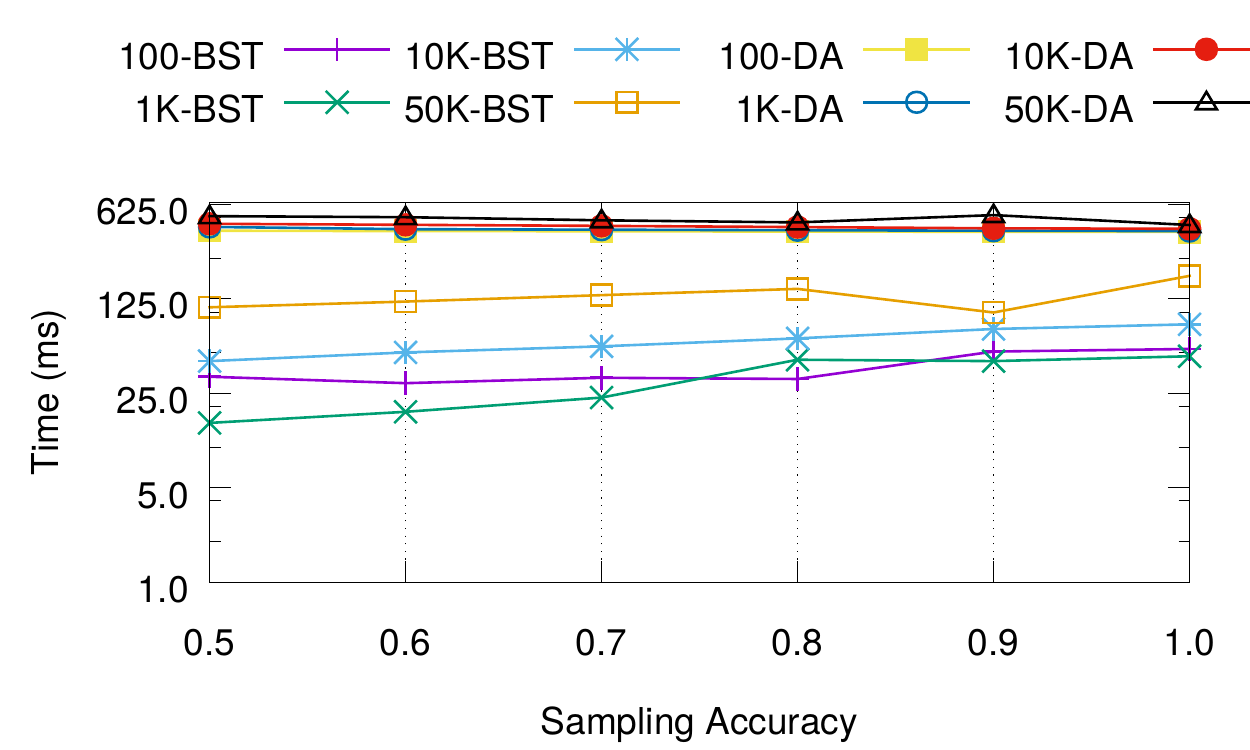}
			\label{fig:S_ts_7}
		}
	\caption{Avg. time taken for sampling with $M = 10^7$}
	\label{fig:S_t_7}
\end{figure}

\begin{figure}[htb]
\centering
		\subfloat[][Uniformly random query set]{
			\includegraphics[width=0.6\columnwidth]{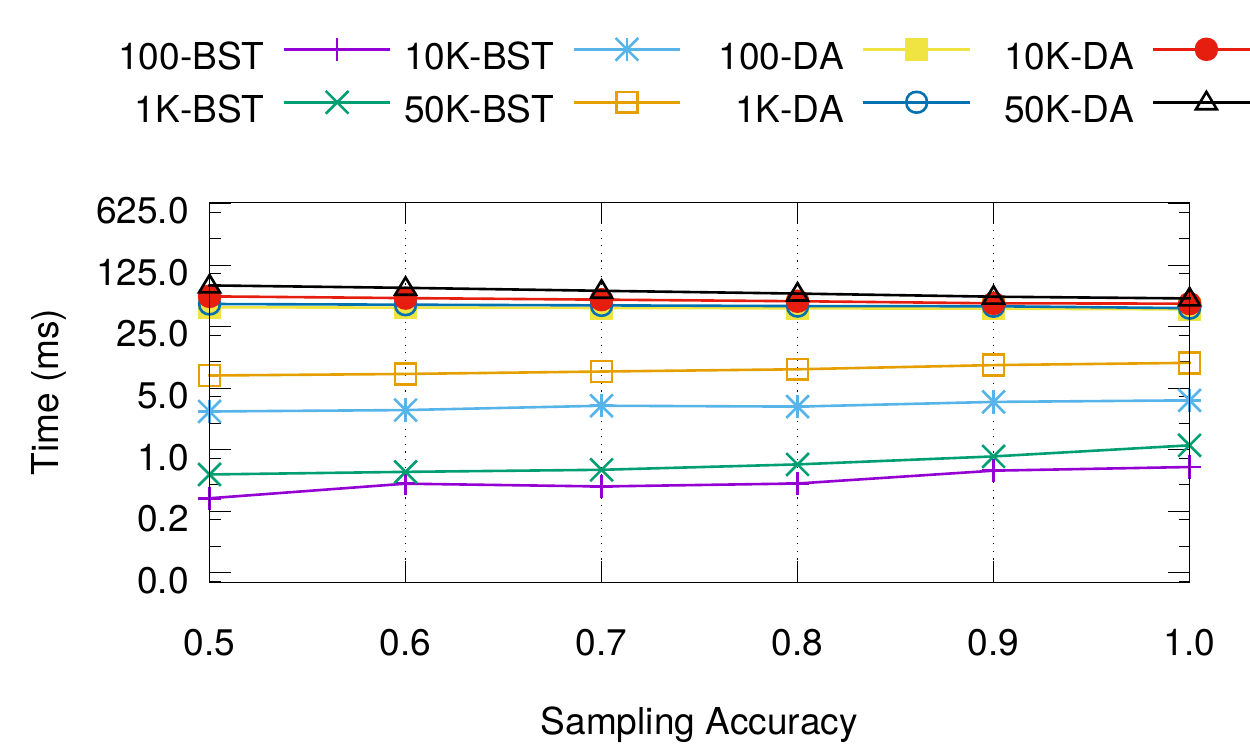}
			\label{fig:S_tu_6}
		}\\
		\subfloat[][Clustered query set]{
			\includegraphics[width=0.6\columnwidth]{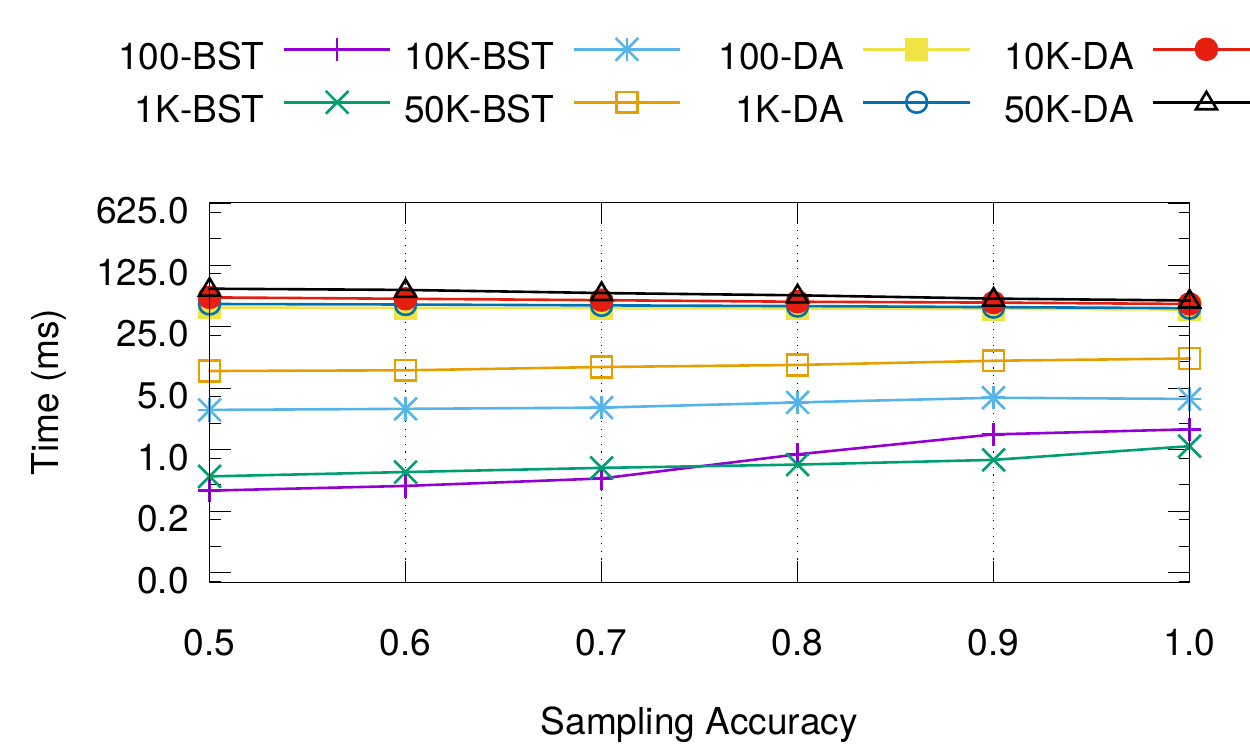}
			\label{fig:S_ts_6}
		}
	\caption{Avg. time taken for sampling with $M = 10^6$}
	\label{fig:S_t_6}
\end{figure}

\begin{figure}[htb]
	\centering
	\includegraphics[width=0.6\columnwidth]{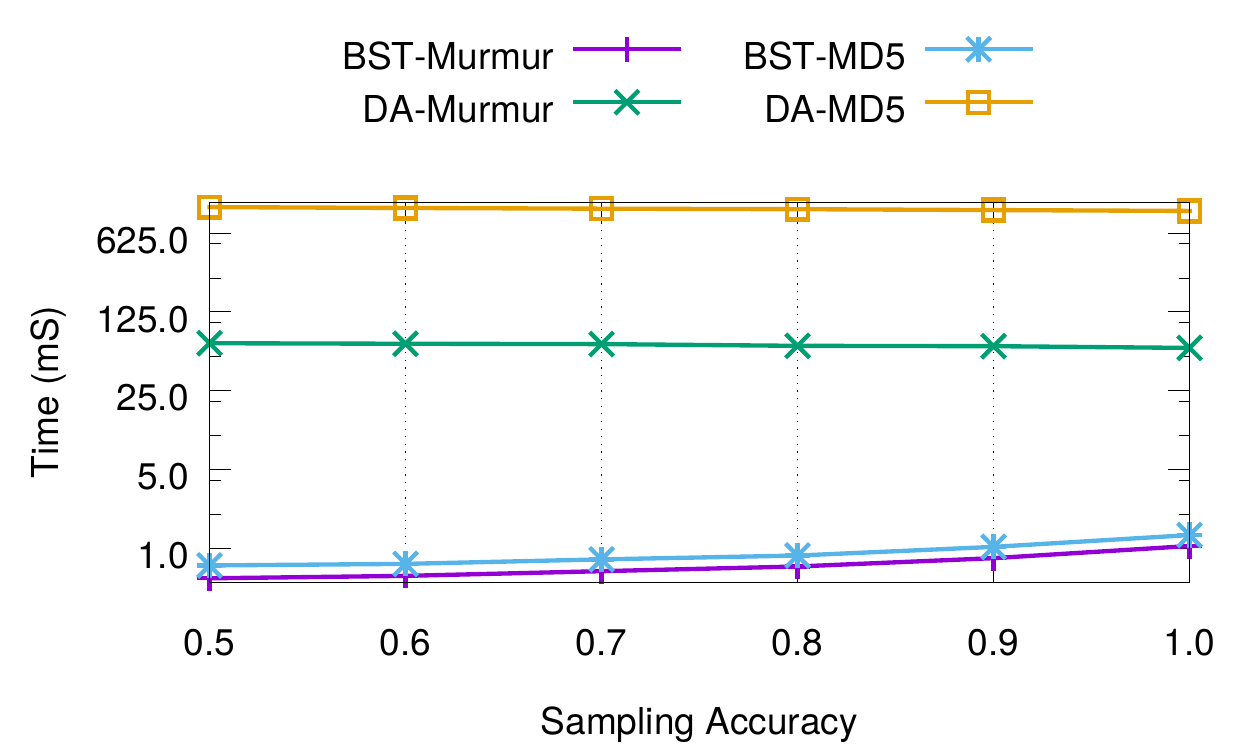}
	\caption{Effect of different hash function families on performance}
	\label{fig:hash}
\end{figure}

Another implementation choice which can significantly affect the performance numbers is that of hash functions. Figure \ref{fig:hash} shows the time taken to generate samples with different hash function families. 

The Dictionary Attack suffers most when the cost of computing the hash function increases -- for instance with MD5-based hash functions the performance goes down by almost an order of magnitude. On the other hand, the \bst{} sampling procedure defers membership queries to lower levels of the tree, by which time most of the tree has already been pruned from the search. When using fast hash functions like Murmur3 or Simple, \bst{} automatically leverages their efficiency to reduce the overall time taken. 

\begin{table}[ht]
\centering
\begin{tabular}{|c|c|c|c|c|}
\hline
Accuracy & $m$    & Depth & $\leafsize$ & Memory$=m*$ \#nodes \\ \hline
0.5 & 28465  & 10 & 976   & 3.467 \\ \hline
0.6 & 32808  & 10 & 976   & 3.997 \\ \hline
0.7 & 38259  & 10 & 976   & 2.326 \\ \hline
0.8 & 46000  & 9  & 1953  & 2.706 \\ \hline
0.9 & 60870  & 9  & 1953  & 3.7   \\ \hline
1   & 137230 & 6  & 15625 & 1.03  \\ \hline
\end{tabular}
\caption{Various parameters settings in Bloom Sample Tree implementation for $n=10^3$ and $M=10^6$ (Memory in MBs) }
\label{tab:MemReq6}
\end{table}

\begin{table}[ht]
\centering
\begin{tabular}{|c|c|c|c|c|}
\hline
Accuracy & $m$    & Depth & $\leafsize$ & Memory$=m*$ \#nodes \\ \hline
0.5      & 63120  & 13    & 1220        & 61.62               \\ \hline
0.6      & 72475  & 13    & 1220        & 70.75               \\ \hline
0.7      & 84215  & 13    & 1220        & 82.22               \\ \hline
0.8      & 101090 & 13    & 1220        & 98.69               \\ \hline
0.9      & 132933 & 12    & 2441        & 64.87               \\ \hline
1.0      & 297485 & 10    & 9765        & 36.27               \\ \hline
\end{tabular}
\caption{Various parameters settings in Bloom Sample Tree implementation for $n=10^3$ and $M=10^7$ (Memory in MBs) }
\label{tab:MemReq7}
\end{table}

\paragraph*{Memory requirement} Finally, we turn our attention to the amount of memory footprint needed by each method. The memory requirements (in MBs) are shown in Tables \ref{tab:MemReq6} and \ref{tab:MemReq7} when the number of elements in the query set is $n=10^3$. In the Tables \ref{tab:MemReq6} and \ref{tab:MemReq7}, Depth is $\log{M/\leafsize}$ where $\leafsize$ is computed as discussed in Section \ref{sec:summary}, and memory is analytically computed using $m * $ number of nodes in the \bst{}. The memory of the \bst{} thus computed was further affirmed by empirical measurement during program execution. It is evident from this table that memory requirement might actually reduce with increasing accuracy. This is primarily because when the depth of the \bst{} decreases, the total memory occupied reduces. This is in spite of increased Bloom filter sizes, since lower levels have much larger memory footprint than higher ones. Since memory can reduce with increasing accuracy, the overall trade-off is between accuracy and memory on one hand, and running time on the other.

While the \bst{} allows for very fast sampling, it requires small additional storage than the other methods described in this paper. Moreover, one does not need to store a \bst{} for each possible query set. There is only one \bst{} for a given size of namespace, Bloom filter size, and choice of hash functions. 

\paragraph*{Quality of Sampling} We use the Pearson's chi-squared test, which we briefly describe here, to empirically validate the sample quality. We conduct $T$ sampling rounds from a Bloom filter $\CB$ storing a set $S = (S_1, S_2, \ldots S_n)$. Now, $\forall i \in [1,n]$, let $o_i$ be the number of times element $S_i$ is sampled. Similarly, let $e_i$ be the expected number of times element $S_i$ should be sampled. Our null hypothesis, $H_0$, is that the sampling is uniform, or restated, that $\forall i \in [1,n]$, $e_i = \frac{T}{n}$. The goal of the chi-squared test would be to see if the null hypothesis should be rejected given the observations $o_i$. We define a random variable $Q = \sum_{i=1}^{n} \frac{(o_i - e_i)^2}{e_i}$. Then $Q$ follows a $\chi^2$ distribution with $n - 1$ degrees of freedom. Given an observation $(o_1, o_2, \ldots o_n)$, we compute the value of $Q$. Let this value be $q$. The p-value is defined as $P(Q \geq q | H_0)$. Clearly, smaller the p-value, higher is the value of $q$, indicating greater deviation from the expectation. In other words, a smaller p-value indicates that the observation has lesser support for $H_0$. If the p-value falls below a threshold $s$, known as the \emph{significance level}, then $H_0$ is rejected, otherwise it is not. The significance level is typically set around $0.05$. We set it to  $0.08$ and use $T = 130 \times n$, the recommended sample size for this significance level \cite{stamatis2002six}. For $M = 10^6$, the p-values thus obtained are reported for sets of different sizes in Table \ref{tab:p-values}. All of the entries in this table are $ > 0.08$, and therefore the null hypothesis is not rejected in any case. For higher values of accuracy, it is clear that the distribution of the elements is close to the uniform distribution.

\paragraph*{Accuracy} While the value of $m$ was determined based on accuracy, we verified the accuracy obtained from the sampling process using the expression for measured accuracy in Table \ref{tab:params}. For all cases, measured accuracy was found to be close to the expected value. Table \ref{tab:MeasuredAccuracies} shows measured accuracy values for $n=1000$.

\emph{Creation Time:} In Table \ref{tab:bstCreation}, we measure the time taken to create the \bst{} for different sizes of the namespace and desired accuracy. Note, for the case of $M=10^5$, with increased accuracy from $0.5$ to $0.6$, the required value of $m$ increases, resulting in a decreased depth of the \bst{} (Section \ref{sec:summary}), and consequently lower creation time. For $M = 10^7$ and desired accuracy $0.9$, creation takes $< 2.5s$. In the application scenarios the \bst{} works with, the number of sets that must be constructed in Bloom filters is assumed to be massive, and changing Bloom filter parameters requires creating each of the sets in the database again which is prohibitively expensive, overshadowing the time taken to create the \bst{}.
\clearpage
\begin{table}[]
\centering

\begin{tabular}{c|cccc|cccc|cccc}
\toprule
a   & \multicolumn{4}{c}{M=10\textasciicircum 5} & \multicolumn{4}{c}{M=10\textasciicircum 6} & \multicolumn{4}{c}{M=10\textasciicircum 7} \\
\midrule
    & m        & \#Levels    & STime   & UTime   & m        & \#Levels    & STime   & Utime   & m        & \#Levels   & Stime    & Utime   \\
\midrule
0.5 & 12317    & 4           & 0.12    & 18.4    & 28464    & 6           & 0.96    & 156.5   & 63119    & 8          & 78.72    & 1603.8  \\
0.6 & 14334    & 3           & 0.04    & 14.4    & 32807    & 6           & 0.68    & 161.7   & 72474    & 8          & 223.52   & 1625    \\
0.7 & 16863    & 3           & 0.28    & 14.8    & 38258    & 6           & 0.64    & 165.2   & 84214    & 8          & 311      & 1616.3  \\
0.8 & 20494    & 3           & 0.12    & 15.6    & 46090    & 5           & 0.72    & 158.9   & 101088   & 8          & 426.52   & 1646.3  \\
0.9 & 27340    & 3           & 0.12    & 15.9    & 60869    & 5           & 0.56    & 163.4   & 132932   & 8          & 590.64   & 1720.2 \\
\bottomrule
\end{tabular}
\caption{System and User Time taken (in mS) to create \bst{} for different values of $M$ and desired accuracy}
\label{tab:bstCreation}
\end{table}

\subsection{Reconstruction Experiments}

\begin{figure}[htb]
	\centering
	\subfloat[][Uniformly Random Query Set] {
		\includegraphics[width=0.6\columnwidth]{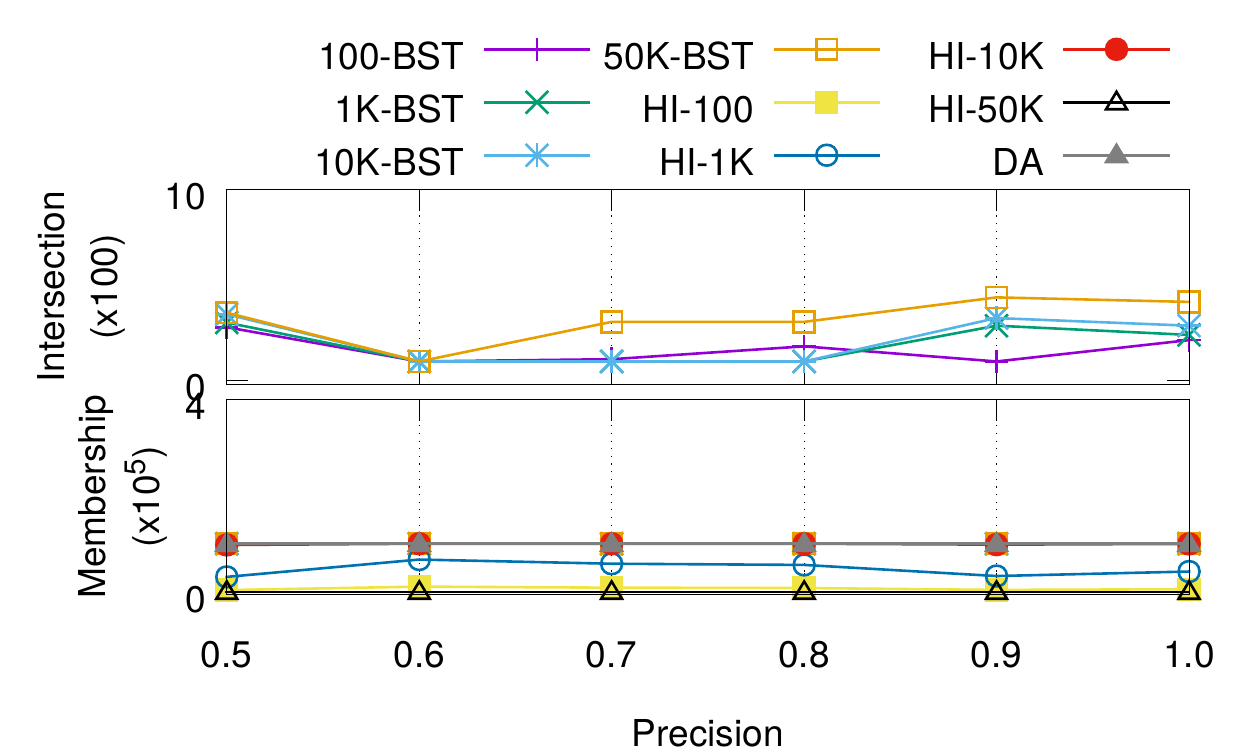}
		\label{fig:R_u_5}
	}\\
	\subfloat[][Clustered Query Set] {
		\includegraphics[width=0.6\columnwidth]{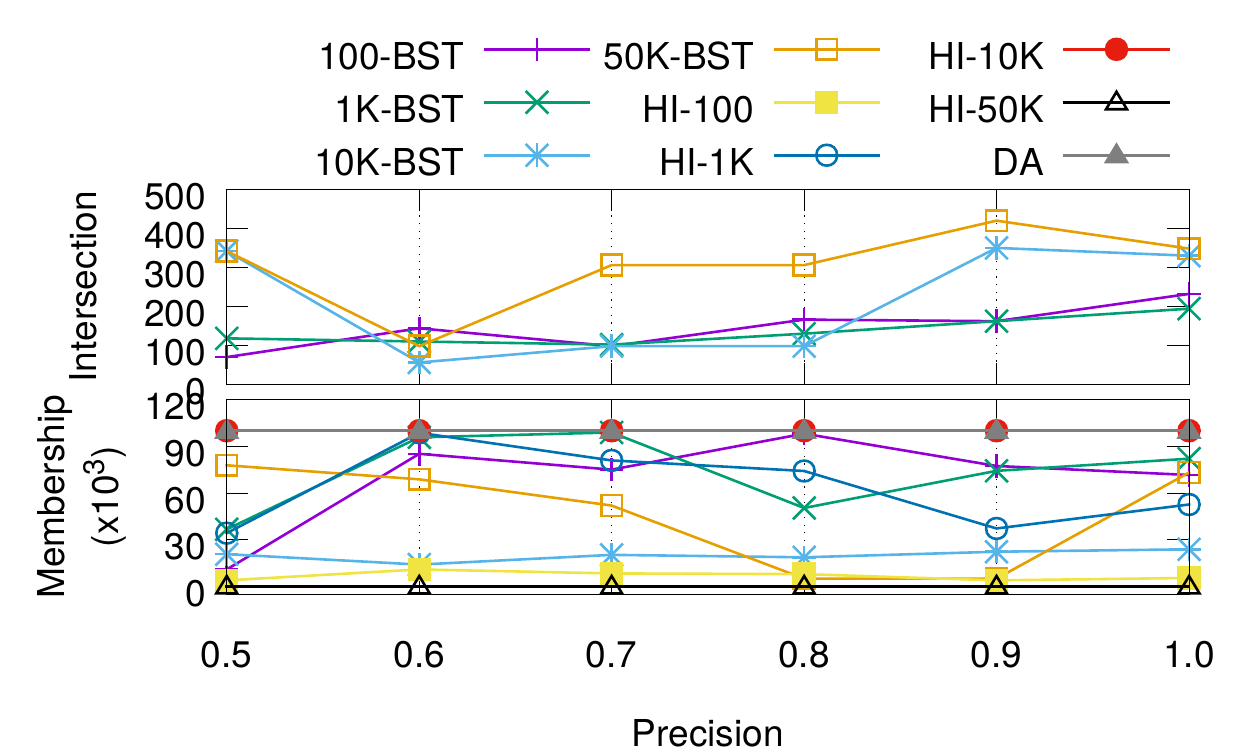}
		\label{fig:R_s_5}
	}
	\caption{Avg. No. of operations in reconstructing for $M=10^5$}
	\label{fig:R_us_5}
\end{figure}

\begin{figure}[htb]
	\centering
	\subfloat[][Uniformly Random Query Set] {
		\includegraphics[width=0.6\columnwidth]{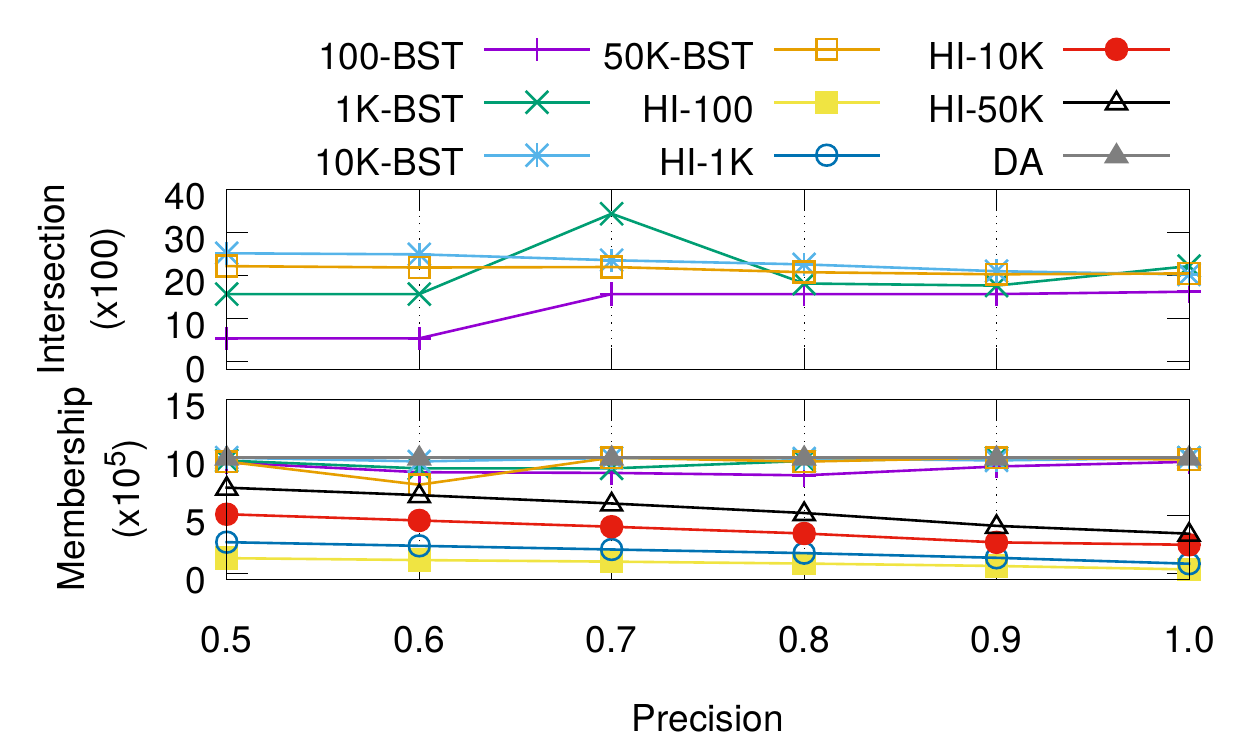}
		\label{fig:R_u_6}
	}\\
	\subfloat[][Clustered Query Set] {
		\includegraphics[width=0.6\columnwidth]{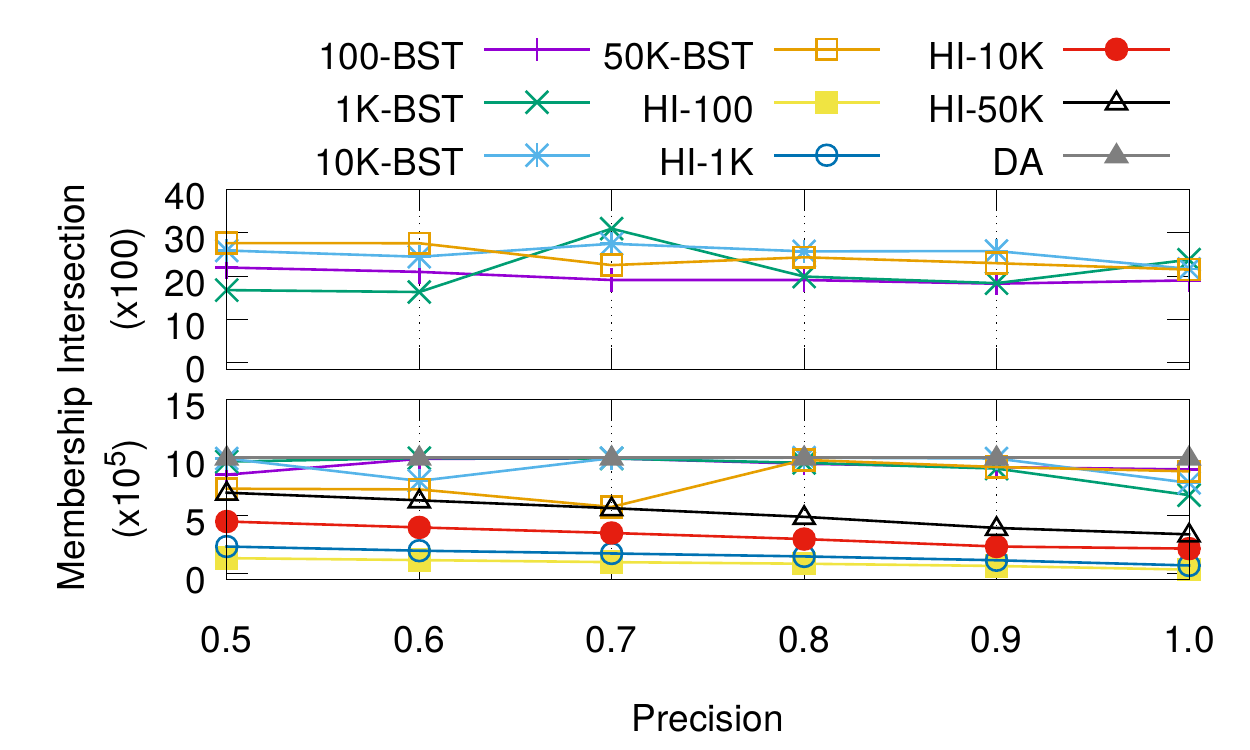}
		\label{fig:R_s_6}
	}
	\caption{Avg. No. of operations in reconstructing for $M=10^6$}
	\label{fig:R_us_6}
\end{figure}

\begin{figure}[htb]
	\centering
	\subfloat[][Uniformly Random Query Set] {
		\includegraphics[width=0.6\columnwidth]{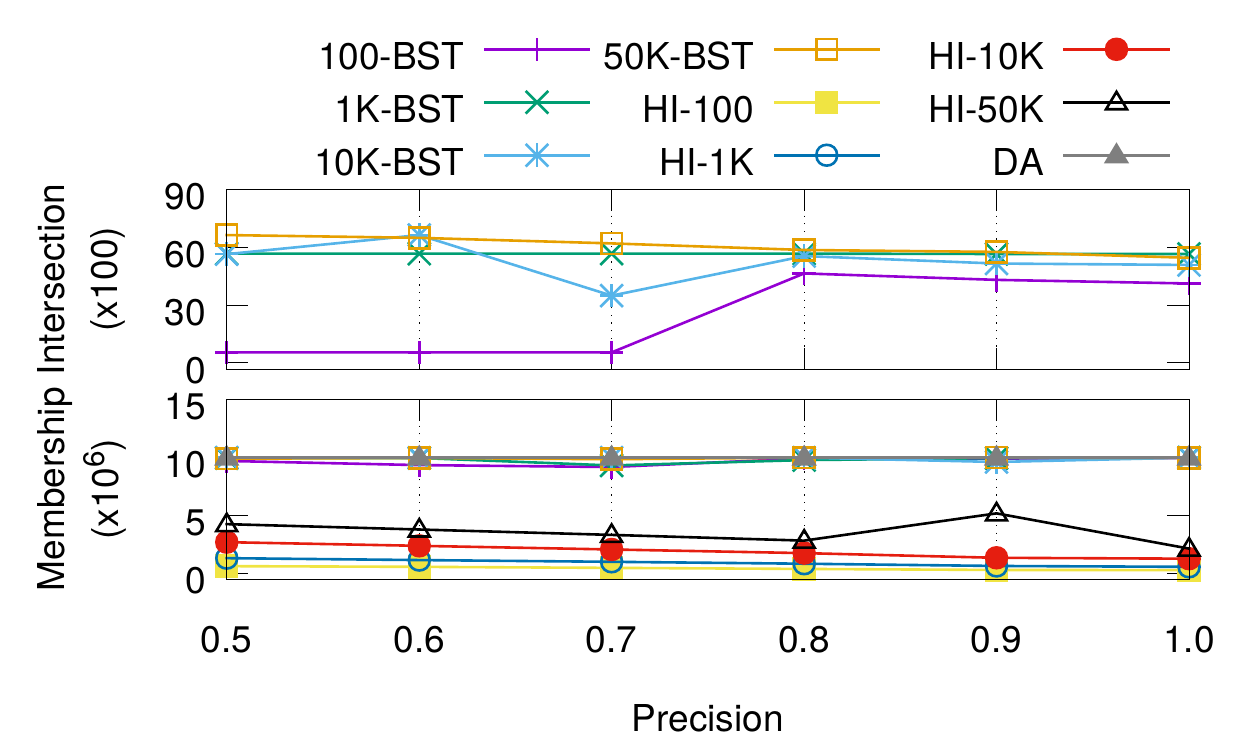}
		\label{fig:R_u_7}
	}\\
	\subfloat[][Clustered Query Set] {
		\includegraphics[width=0.6\columnwidth]{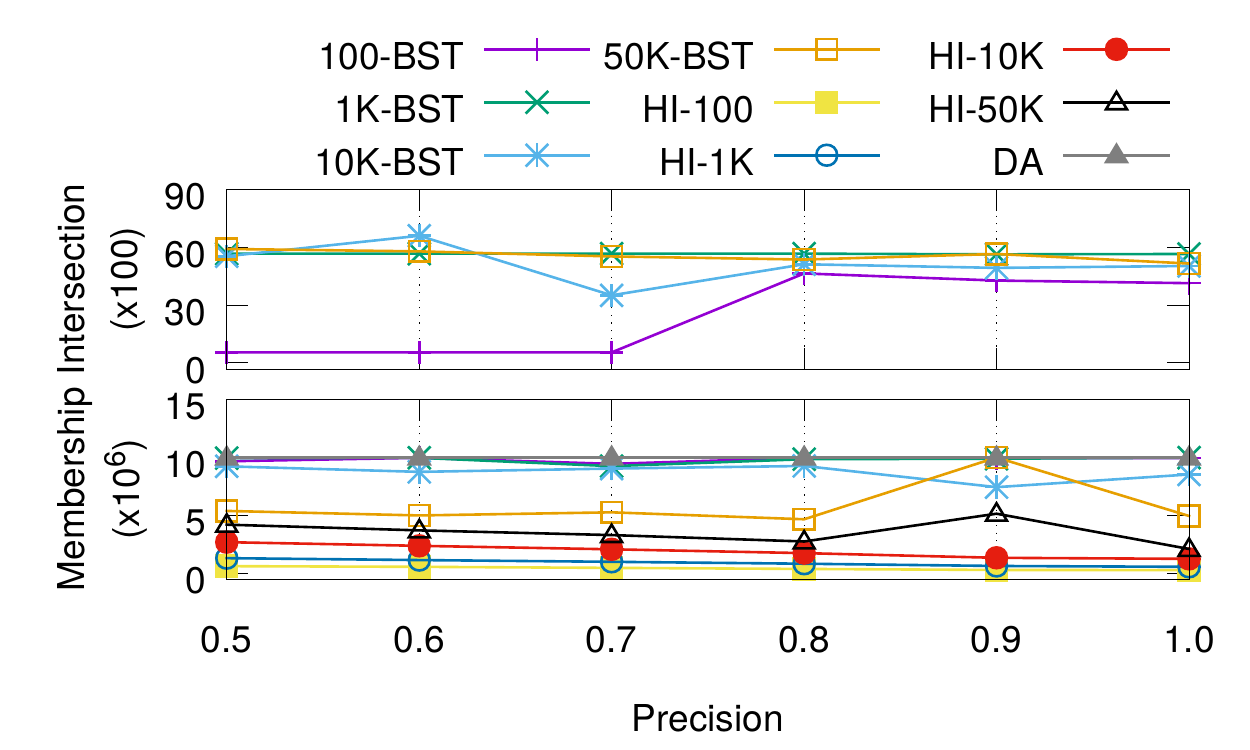}
		\label{fig:R_s_7}
	}
	\caption{Avg. No. of operations in reconstructing for $M=10^7$}
	\label{fig:R_us_7}
\end{figure}

The setup of the reconstruction experiments follow that of the sampling experiments only adding \emph{HashInvert} as a baseline.
Figures \ref{fig:R_us_6} and \ref{fig:R_us_7} show the number of intersections and set membership queries to reconstruct sets which are uniformly random and clustered, drawn from namespaces of size $M = 10^6$ and $M = 10^7$ respectively.

For the number of intersections with sampling accuracy, we see a trend that is similar to the ones in the sampling experiments, and for the same reasons. One may note that the HashInvert procedure performs more membership queries than the \bst{}, but fewer than the Dictionary Attack. Despite this, the overall cost for HashInvert is the most as can be seen in Figures \ref{fig:R_tus_6} and \ref{fig:R_tus_7}, which show the time taken for reconstruction. The overall cost for HashInvert essentially depends on the number of set or reset bits in the Bloom filter. If the Bloom filter is extremely dense, then reconstructing with the help of only reset bits efficiently reconstructs the set, whereas if it is very sparse, then one can reconstruct using the set bits. However, HashInvert is inefficient if neither of these cases apply, as is evident from the line for 'HI-10K', which sets about $50 \%$ of the bits in the Bloom Filter. The cause for this is the fact that HashInvert iterates through an inverted set for each set or reset bit in the Bloom filter. Since some of these values may already have been checked, it does save some membership queries. However, given that the membership query is very fast for simple hash functions, this does not directly translate into smaller running times.

\begin{figure}[htb]
\centering
		\subfloat[][Uniform Random Query Set]{
			\includegraphics [width=0.6\columnwidth]{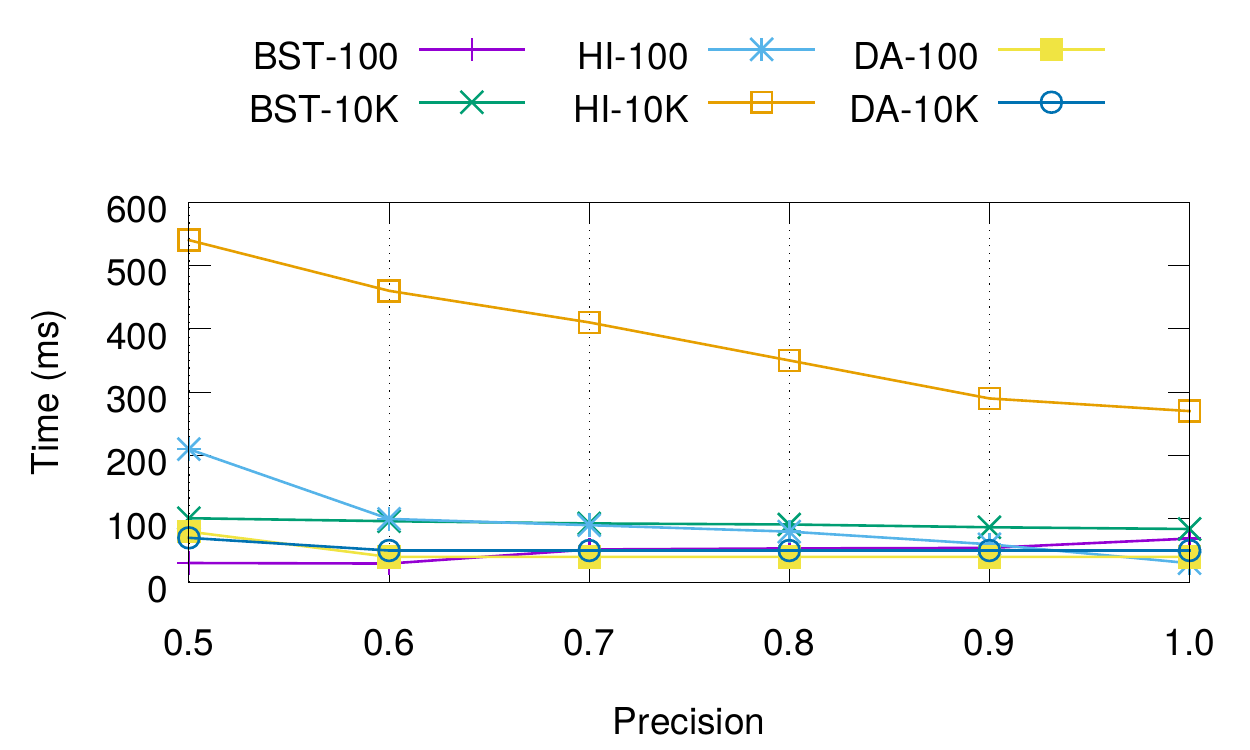}
			\label{fig:R_tu_6}
		}\\
	\subfloat[][Clustered Query Set]{
		\includegraphics[width=0.6\columnwidth]{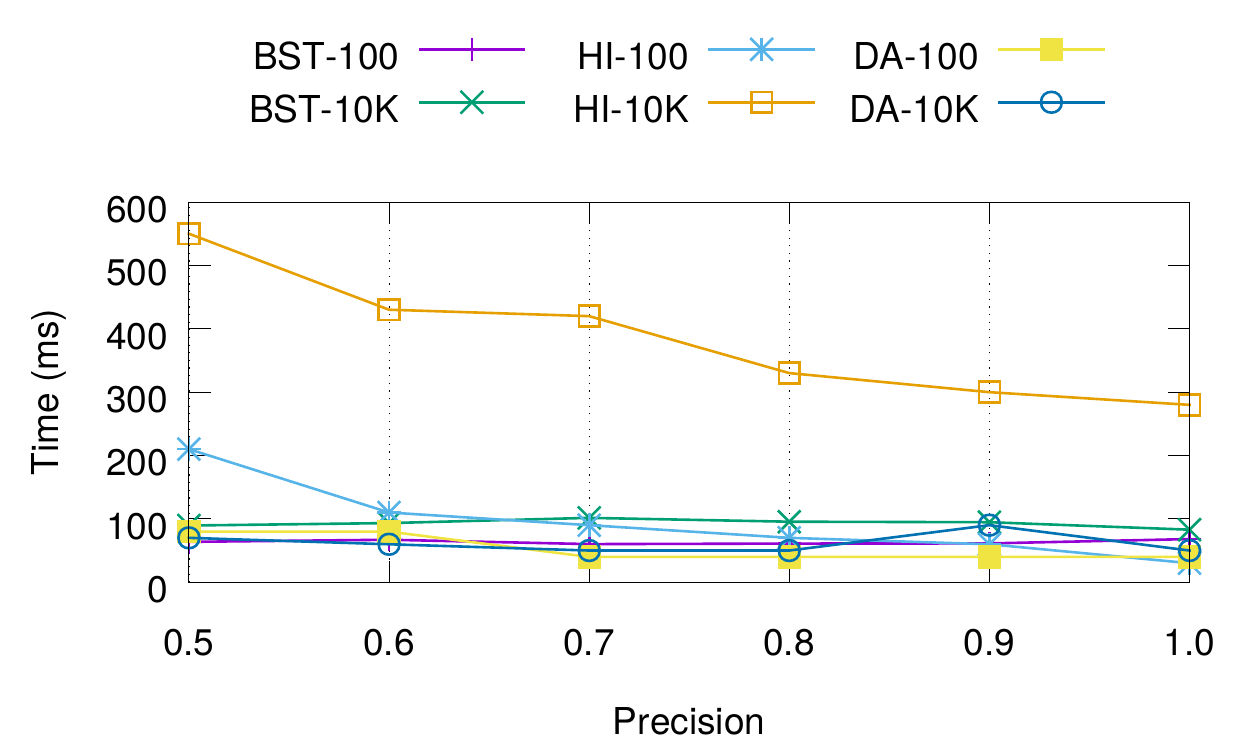}
		\label{fig:R_ts_6}
	}
	\caption{Avg. time taken for reconstruction with $M = 10^6$ for uniformly random and clustered query sets.}
	\label{fig:R_tus_6}
\end{figure}
\begin{figure}[htb]
\centering
		\subfloat[][Uniform Random Query Set]{
			\includegraphics [width=0.6\columnwidth]{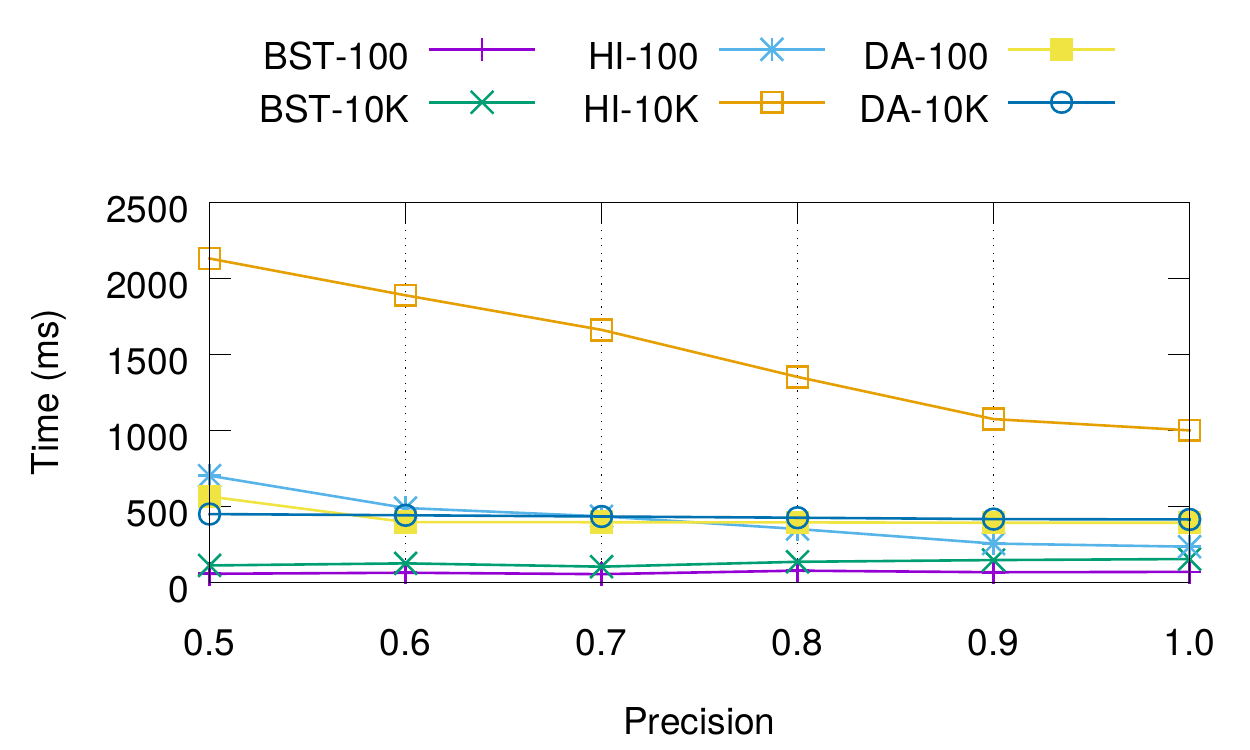}
			\label{fig:R_tu_7}
		}\\
	\subfloat[][Clustered Query Set]{
		\includegraphics[width=0.6\columnwidth]{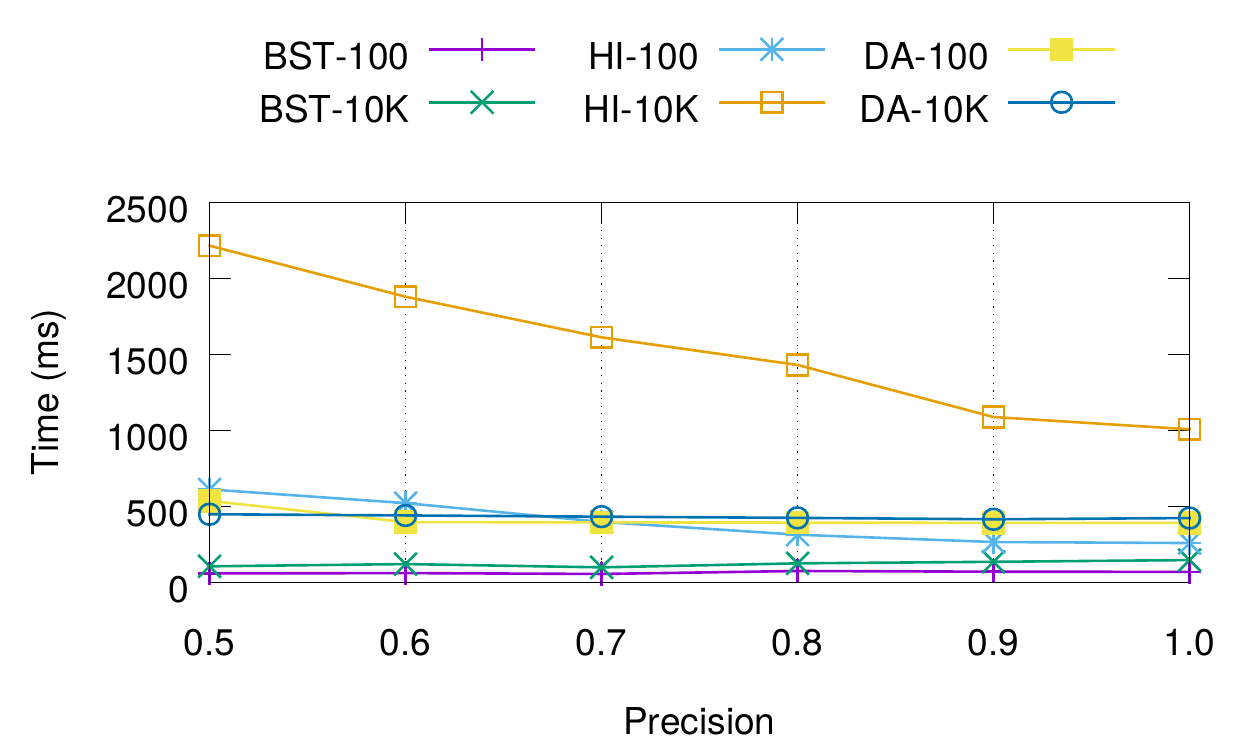}
		\label{fig:R_ts_7}
	}
	\caption{Avg. time taken for reconstruction with $M = 10^7$ for uniformly random and clustered query sets.}
	\label{fig:R_tus_7}
\end{figure}

\begin{table}[]
\centering
\begin{tabular}{|c|c|c|c|c|}
\hline
Accuracy / $n$ & 100  & 1K   & 10K  & 50K  \\ \hline
0.5            & 1    & 0.99 & 0.52 & 0.78 \\ \hline
0.6            & 1    & 0.92 & 0.75 & 0.88 \\ \hline
0.7            & 0.99 & 0.15 & 0.87 & 0.63 \\ \hline
0.8            & 0.93 & 0.49 & 0.51 & 0.12 \\ \hline
0.9            & 0.93 & 0.75 & 0.28 & 0.47 \\ \hline
1.0            & 0.84 & 0.48 & 0.43 & 0.64 \\ \hline
\end{tabular}
\caption{p-values for $M=10^6$}
\label{tab:p-values}
\end{table}

\begin{table}[]
\centering
\begin{tabular}{|c|c|c|c|}
\hline
Accuracy / $M$ & $10^5$ & $10^6$ & $10^7$ \\ \hline
0.5            & 0.522  & 0.497  & 0.535  \\ \hline
0.6            & 0.692  & 0.621  & 0.591  \\ \hline
0.7            & 0.710  & 0.691  & 0.696  \\ \hline
0.8            & 0.823  & 0.793  & 0.810  \\ \hline
0.9            & 0.921  & 0.907  & 0.906  \\ \hline
1.0            & 0.970  & 0.997  & 0.948  \\ \hline
\end{tabular}
\caption{Measured Accuracies for Uniform query sets of size $n=10^3$}
\label{tab:MeasuredAccuracies}
\end{table}

\section{Experiments with Real-world Data with Low-Occupancy Namespace}
\label{sec:expt-dynamic}
So far we presented results for the settings when the namespace is a contiguous and fixed. Now we turn our attention to more practical settings where the size of the namespace we need to handle is only a small fraction of a much larger domain and potentially spread throughout it.

\subsection{Setup}

\paragraph*{Dataset} We made use of a 34-day Twitter crawl consisting of 144 million tweets. There are a total of 7.2 million user ids in this tweet set, but they are distributed in a namespace of $\left[0, 2\times10^9\right]$ (a little over $2.2$ billion). 

\paragraph*{Varying the namespace fractions} Note that even though there are only 7.2 million unique ids in our dataset, they could be distributed across the entire namespace of 2.2 billion. Suppose, for example, we built a \bst with 256 leaves -- that is, the range of 2.2 billion is effectively divided into 256 equal-sized ranges (of which some could be empty depending on the distribution of the unique ids). From this hypothetical \bst, we construct namespaces of different namespace fractions as follows:
	\begin{itemize}
		\item {\bf Uniform Namespace:} Following our example above, suppose we want to construct a namespace of namespace fraction 0.2, we \emph{uniformly} sample 52 of 256 leaves. This gives us a set of ranges, the union of which only occupy 0.2 fraction of the total namespace.
		\item {\bf Clustered Namespace:} Again, for a namespace fraction of 0.2, we need to sample 52 of 256 leaves, but in a \emph{clustered} way. We use the same technique as explained in Section \ref{sec:expt} (in that case, we were generating clustered query sets).
	\end{itemize}

We fixed the desired accuracy, as discussed in section \ref{sec:summary}, at 0.8. Therefore, our hypothetical \bst{} has a depth of 7, with a Bloom filter size $m = 1.2 \times 10^6$. Correspondingly, the pruned-\bst{} has the same depth and Bloom filter size, but the number of nodes (and therefore the space occupancy) is much smaller.

\paragraph*{Query Bloom filters} We identified $24,000$ unique hash tags that occurred at least $1,000$ times in our dataset. The sets of users tweeting a particular hashtag is used to construct a query Bloom filter. We therefore constructed 24,000 query Bloom filters. However, when experimenting with varying namespace fractions, we simply ignore ids which do not belong in the namespace currently under consideration and construct query Bloom filters without them.

\paragraph*{Metrics} We report the following metrics.
\begin{itemize}
	\item {\bf Average Time taken.} At each namespace fraction, we run 1,000 sampling rounds on randomly chosen query Bloom filters and report the average time taken to generate a sample.
	\item {\bf Memory.} The Pruned-\bst{} occupies much less space than the full \bst{}. We report on space usage at each namespace fraction.
	\item {\bf Accuracy.} While the value of $m$, the Bloom filter size, was based on a desired accuracy for the \bst{}, the actual accuracy in a Pruned-\bst{} is expected to be better, since only those elements which occupy the namespace are stored. We report this accuracy for various namespace fractions.
\end{itemize}

\subsection{Sampling Experiments}

\paragraph*{Average time taken}
Figure \ref{fig:time} shows the average time taken to generate samples from our query Bloom filters. At namespace fractions less than $0.1$, the time taken is an order of magnitude smaller than at full namespace occupancy. It is also expected that the sampling time in case of the clustered namespace is smaller, since more leaves share common ancestors and there are far less paths in the \bst{} for the sampling algorithm to follow.
The Dictionary Attack requires $100$ seconds on average for one sample to be drawn. This is natural since the size of the namespace is extremely large in this case. As a result, we have not included the result of DA in Figure \ref{fig:time} to ensure that the finer variations in the sampling time taken for random and clustered namespaces are clearly visible.

\begin{figure}
	\centering
	\includegraphics[width=0.6\columnwidth]{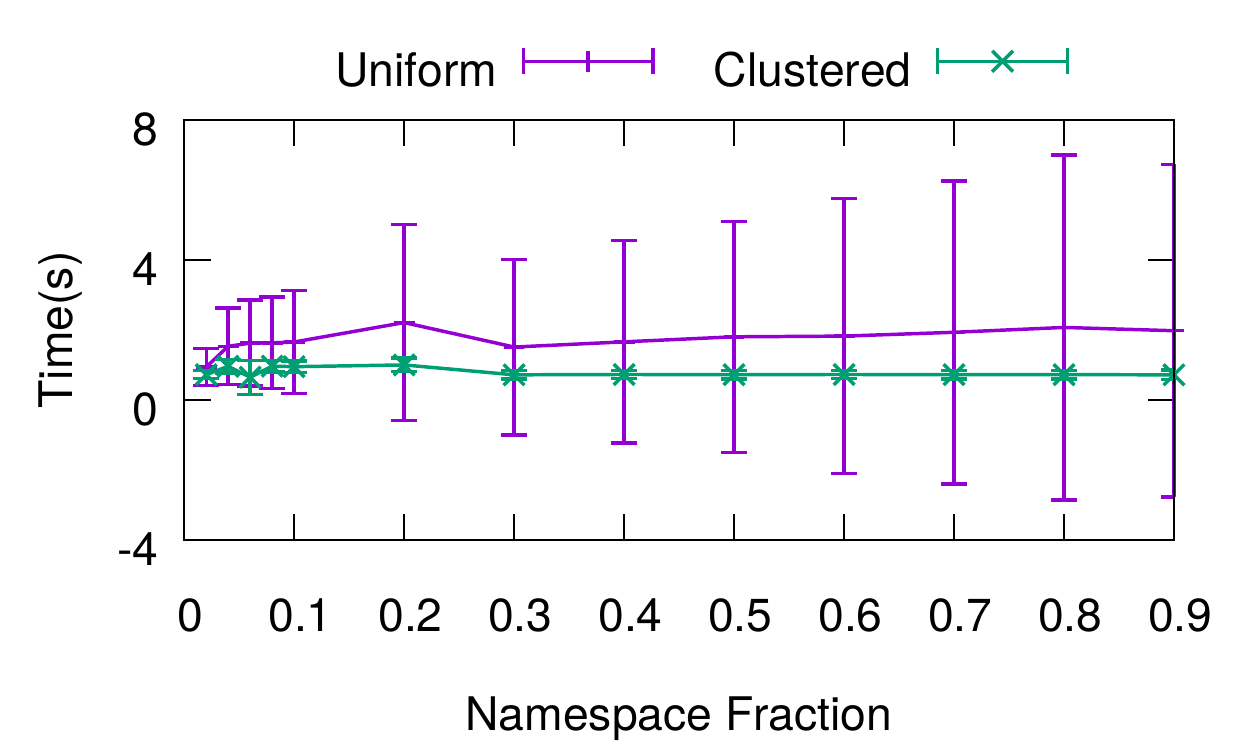}
	\caption{Time taken to generate a uniform sample for varying namespace fractions}
	\label{fig:time}
\end{figure}

\paragraph*{Memory}
Figure \ref{fig:memory} shows the memory usage at varying namespace fractions. Note that, if we built the full \bst{} for a namespace of 2.2 billion, the memory required would be approximately $36 MB$. In contrast, at a lower namespace fraction of $0.5$, the memory usage of the \bst{} is about $ 71 \%$ for the uniform case, and much lower at $21.7 \%$ for the clustered case. For the same reason as for sampling time, we expect the memory requirement of the \bst{} to be smaller for a clustered namespace.

\begin{figure}
	\centering
	\includegraphics[width=0.6\columnwidth]{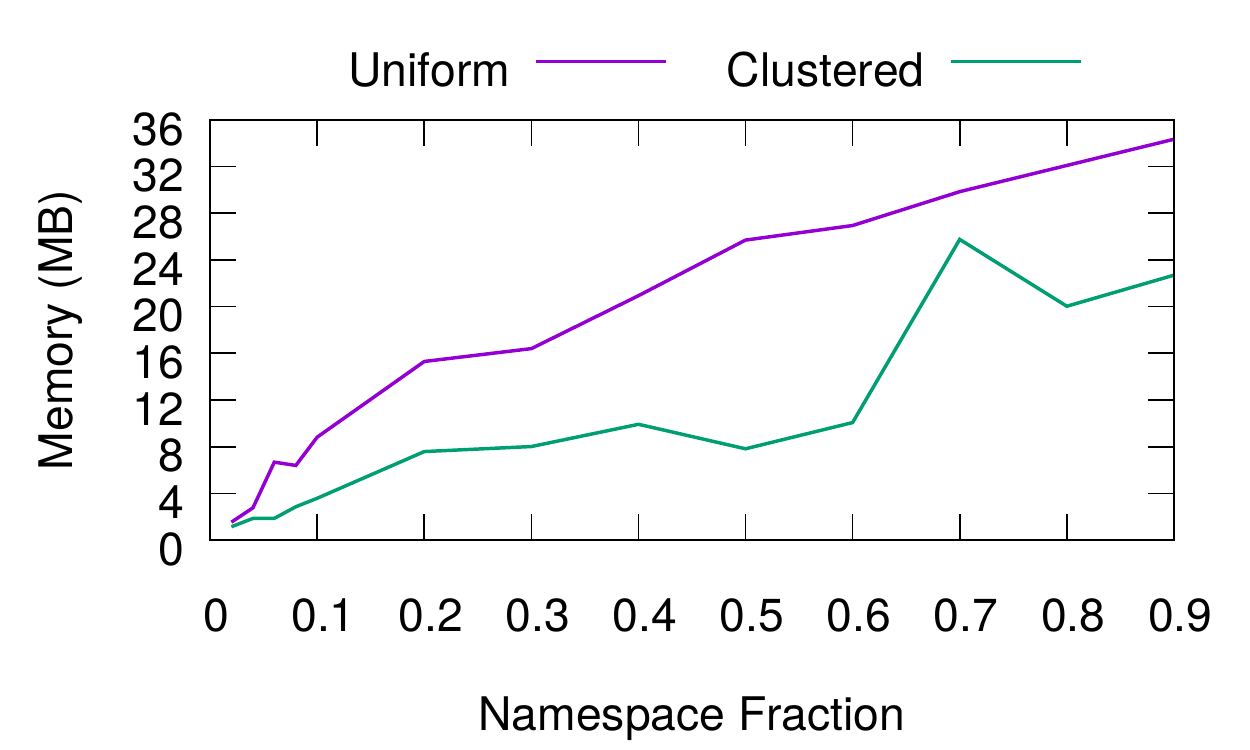}
	\caption{Memory usage at varying namespace fractions}
	\label{fig:memory}
\end{figure}

\paragraph*{Accuracy}
Figure \ref{fig:accuracy} shows the sampling accuracy at various namespace fractions. Recall that we had optimized the \bst{} for an accuracy of 0.8. But, with our Pruned-\bst, we uniformly see a higher accuracy. Accuracy depends on the size of the namespace, as mentioned in section \ref{sec:summary}, and the size of the effective namespace at a lower namespace fraction is smaller. This shows that the \bst{} is capable of producing higher accuracy results when the overall namespace is large but the actually occupied effective namespace is small.

\begin{figure}
	\centering
	\includegraphics[width=0.6\columnwidth]{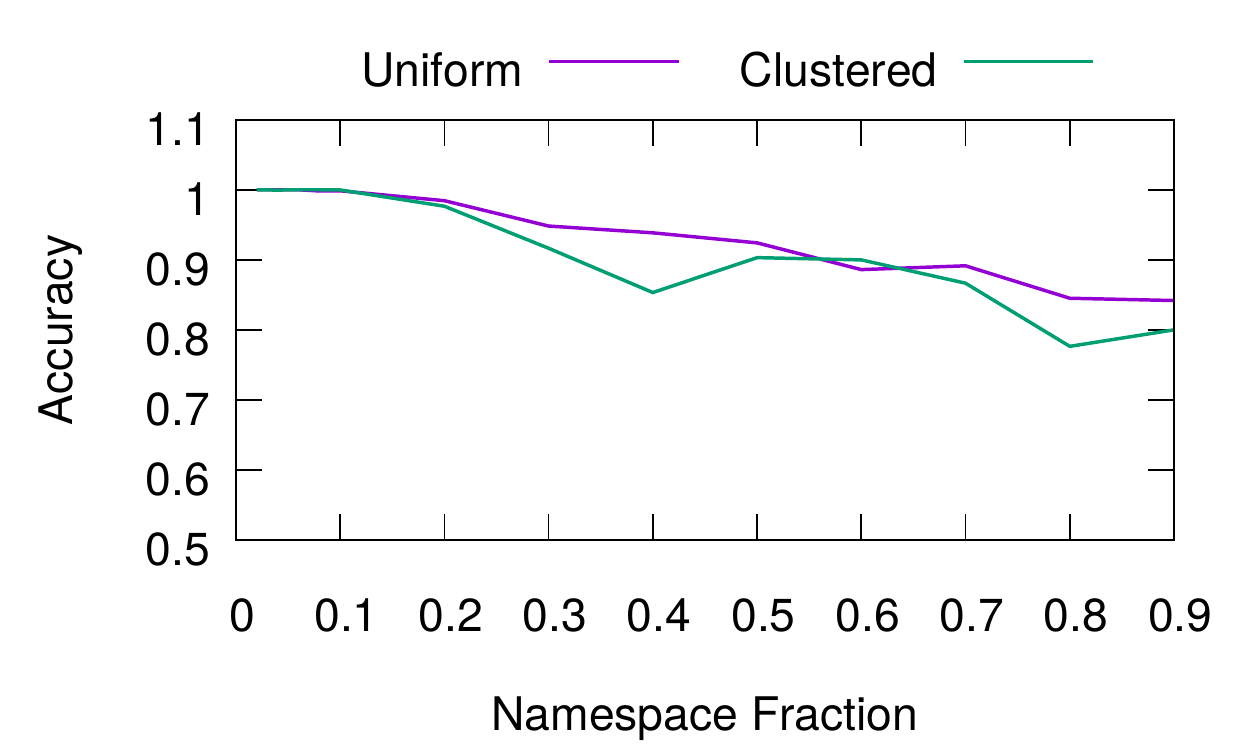}
	\caption{Sampling accuracy at varying namespace fractions}
	\label{fig:accuracy}
\end{figure}

\section{Conclusions}

In this paper we described an efficient method to do sampling and reconstruction of sets stored in Bloom filters. In particular, we described the \bst{} data structure and analyzed its properties both theoretically and experimentally. We compared our technique to the brute force approach (Dictionary Attack) as well as HashInvert (useful when using invertible hash functions to reconstruct sets). An extensive evaluation of our algorithm in various settings demonstrated  its wide applicability and significant advantages. 


\end{document}